\documentclass[12pt, reqno]{amsart}
\usepackage{amsmath, amssymb}
\usepackage{graphicx, color, wrapfig, framed}
\usepackage[height=22cm, width=15.5cm, hmarginratio={1:1}]{geometry}

\newcommand{\beq}{\begin{equation}}
\newcommand{\eeq}{\end{equation}}
\newcommand{\bal}{\begin{align}}
\newcommand{\eal}{\end{align}}
\newcommand{\ee}{\epsilon}
\newcommand{\e}{\epsilon}
\newcommand{\p}{\partial}

\newcommand{\bt}{{\bf t}}
\newcommand{\CC}{\mathbb{C}}
\newcommand{\br}{{\bf r}}

\newcommand{\bigs}{{\bf s}}
\newcommand{\bigzero}{{\bf 0}}

\newtheorem{pro}{Proposition}
\newtheorem{lem}{Lemma}
\newtheorem{Def}{Definition}
\newtheorem{cor}{Corollary}
\newtheorem{thm}{Theorem}
\newtheorem{rmk}{Remark}

\begin{document}

\title{Galilean symmetry of the KdV hierarchy}
\author{Jianghao Xu, Di Yang}
\begin{abstract}
By solving the infinitesimal Galilean symmetry for the KdV hierarchy, 
we obtain an explicit expression for the corresponding one-parameter Lie group, which 
we call the {\it Galilean symmetry} of the KdV hierarchy. 
As an application, we establish an explicit relationship between  
the {\it non-abelian Born--Infeld partition function} and the {\it generalized Br\'ezin--Gross--Witten partition function}.
\end{abstract}
\maketitle
	
\tableofcontents
	
\section{Introduction and statements of the main results}
	
	The KdV (Korteweg--de Vries) equation
	\begin{equation}\label{KdV}
		\frac{\p u}{\p t}=u\frac{\p u}{\p X}+\frac{\e^2}{12}\frac{\p^3 u}{\p X^3}
	\end{equation}
	is an evolutionary partial differential equation, discovered in the study of shallow water waves in the late of the 19th century. 
	In this paper, like in~\cite{DYZ18} we are interested in solutions to the KdV equation in the ring $V[[t,\e]]$, 
	where $V$ is a given ring of functions of~$X$ closed under $\p/\p X=:\p_X$. For simplicity we take $V=\CC[[X]]$.
	A solution $u(X,t;\e)$ in $V[[t,\e]]$ to the KdV equation is uniquely specified by 
	its initial value 
	\begin{equation}\label{initialdatakdveq}
		u(X,0;\e)=:f(X;\e).
	\end{equation}
	
The KdV equation has the well-known {\it Galilean symmetry}\footnote{Here we remind the reader that a symmetry of an equation maps a solution to another solution.}. This means that, if $u(X,t;\e)$ is a solution to the KdV equation \eqref{KdV}, then
	\begin{equation}\label{KdV-GalileanTransformation}
		\tilde{u}(X,t;\e;q):=u(X+qt,t;\e)+q
	\end{equation}
	is also a solution to \eqref{KdV}.
	We observe that, from the viewpoint of the initial value, the effect of the Galilean symmetry~\eqref{KdV-GalileanTransformation} 
	is simply the following: it changes the initial value by a constant. 
	Indeed, 
	the initial value of $\tilde{u}(X,t;\e;q)$ defined by \eqref{KdV-GalileanTransformation} is given by
	\begin{equation}
		\tilde{u}(X,0;\e;q)=f(X;\e)+q.
	\end{equation}
	
It is well known (see e.g.~\cite{Di}) that the KdV equation admits the following infinite family of infinitesimal symmetries, which can be written as 
	\begin{equation}\label{KdV-hierarchy}
		\frac{\p L}{\p t_k}=\frac{1}{(2k+1)!!}\Bigl[\bigl(L^{\frac{2k+1}{2}}\bigr)_{+},L\Bigr],\quad k\geq 0,
	\end{equation}
where $L:=\e^2\p_{X}^2+2u$
is the so-called Lax operator. Here, $L^{\frac{2k+1}{2}}$ are pseudo-differential operators and $(\cdot)_{+}$ means to take the nonnegative part 
(see e.g.~\cite{Di} for details about pseudo-differential operators). 
Each of the infinitesimal symmetries~\eqref{KdV-hierarchy} is an evolutionary equation. In particular, the $k=1$ equation coincides with~\eqref{KdV}, and the $k=0$ equation reads $\p u/\p t_0=\p u/\p X$. We identify $t_0$ with~$X$. Equations~\eqref{KdV-hierarchy} are also infinitesimal symmetries to each other, and all together they form the KdV hierarchy, yielding solutions of the form $u(\bt;\e)$, where $\bt=(t_0,t_1,t_2,\dots)$ 
is the infinite vector of times.
	
	The infinitesimal version of the Galilean symmetry~\eqref{KdV-GalileanTransformation} reads
	\begin{equation}\label{infGali}
		\frac{\p \tilde{u}}{\p q}=t\frac{\p \tilde{u}}{\p X}+1,
	\end{equation} 
	which is called the {\it infinitesimal Galilean symmetry} for the KdV equation~\eqref{KdV}. 
	Sometimes an infinitesimal symmetry 
	is called for short a symmetry.
	It is known that the KdV hierarchy \eqref{KdV-hierarchy} admits the infinitesimal Galilean symmetry \cite{DZ-norm} (see also \cite{CLL, IS})
	\begin{equation}\label{KdV-Galilean-Symmetry}
		\frac{\p \tilde{u}}{\p q}=\sum_{k\geq 0}t_{k+1}\frac{\p \tilde{u}}{\p t_k}+1.
	\end{equation}
The following proposition gives explicitly 
the one-parameter Lie group for~\eqref{KdV-Galilean-Symmetry}.
\begin{pro}\label{prop-galilean-solution}
	For any given power series $u(\bt;\e)\in\mathbb{C}[[\bt,\e]]$, 
		the initial value problem 
		\begin{subequations}\label{IVP-for-solution}
		\begin{align}
				&\frac{\p \tilde{u}(\bt;\e;q)}{\p q}=\sum_{k\geq0} t_{k+1}\frac{\p \tilde{u}(\bt;\e;q)}{\p t_k}+1,\\
				&\tilde{u}(\bt;\e;q=0)=u(\bt;\e)
		\end{align}
		\end{subequations}
		has a unique solution $\tilde{u}(\bt;\e;q)$ in $\mathbb{C}[[\bt,\ee,q]]$, which has the explicit expression
\begin{equation}\label{GalileanTransformationOnSolution}
\tilde{u}(\bt;\e;q)=u\bigl(\bt^{\rm G}(\bt;q);\e\bigr)+q, \quad t^{\rm G}_{n}(\bt;q):=\sum_{k\geq 0}\frac{q^k}{k!}t_{n+k} \, (n\ge0) .
\end{equation}
\end{pro}
	
\noindent We give a proof of this proposition in Section~\ref{S2}. 
Let us mention that an important feature of the linear map $\bt^{\rm G}(\bt;q)$ given in~\eqref{GalileanTransformationOnSolution} is that it is invertible.

The following theorem says that \eqref{GalileanTransformationOnSolution} gives a symmetry of the KdV hierarchy. 
\begin{thm}\label{thm-of-kdv-solution}
If $u(\bt;\e)$ is a solution to the KdV hierarchy, then $\tilde{u}(\bt; \ee; q)$ defined by~\eqref{GalileanTransformationOnSolution} satisfies the KdV hierarchy for arbitrary~$q$.
\end{thm}
\noindent The proof of Theorem~\ref{thm-of-kdv-solution} is in Section~\ref{S2}. 

We call \eqref{GalileanTransformationOnSolution} the {\it Galilean symmetry of the KdV hierarchy}. 
	Just like the observation made above for the KdV equation, 
	the Galilean symmetry \eqref{GalileanTransformationOnSolution} changes the initial value for the KdV hierarchy by a constant.
	
\begin{rmk}
It turns out that 
transformation similar to~$\bt^{\rm G}$ appeared in the WK-BGW correspondence~\cite{YZ2023}. 
	  We will show in this paper that the WK-BGW correspondence of~\cite{YZ2023} can be geometrically interpreted by the Galilean symmetry.
\end{rmk}

A crucial object in the theory of integrable systems is the $\tau$-function 
\cite{DJKM, Di, sato1981soliton} (cf.~\cite{BDY, DYZ18, DZ-norm}). Recall that 
a function $\tau(\bt;\ee)$ is called a {\it $\tau$-function for the KdV hierarchy}, for short a {\it KdV $\tau$-function},  
if it satisfies the following Hirota bilinear equations: 
\begin{equation}\label{hibi}
			\underset{z=\infty}{{\rm res}}\tau(\bt-[z^{-1}];\ee)\tau(\bt'+[z^{-1}];\e)
			e^{\sum\limits_{j\geq 0}\frac{t_{j}-t'_{j}}{\e(2j+1)!!}z^{2j+1}}z^{2p}dz=0,\quad \forall\, \bt,\bt',\,p\ge0.
\end{equation}
Here $[z^{-1}]:=(\frac{\e}{z},\frac{\e}{z^3},\cdots,\frac{\e(2k-1)!!}{z^{2k+1}},\cdots)$.
Let $\tau(\bt;\e)$ be a KdV $\tau$-function. Then, it is easy to verify that 
$\gamma(\bt;\e)\tau(\bt;\e)$, where $\gamma(\bt;\e):=e^{c_{-1}(\e)+\sum_{i\geq 0}c_{i}(\e)t_i}$, 
is also a KdV $\tau$-function 
for arbitrary $c_{-1}(\e),c_0(\e),c_1(\e),c_2(\e)$, \dots.
Put
\beq\label{utauhirota}
u(\bt;\e):=\e^{2}\frac{\p^2 \log\tau(\bt;\e)}{\p t_0^2}.
\eeq
Then $u(\bt;\e)$ satisfies the KdV hierarchy~\eqref{KdV-hierarchy}. 
Note that multiplying~$\tau(\bt;\e)$ by the factor $\gamma(\bt;\e)$ does not change $u(\bt;\e)$. 
Conversely, given a solution $u(\bt;\e)$ to the KdV hierarchy~\eqref{KdV-hierarchy}, with the help of the so-called $\tau$-structure 
(cf.~\cite{BDY, Di, DYZ18, DZ-norm}), one can uniquely determine 
up to multiplying by a factor like $\gamma(\bt;\e)$ a KdV $\tau$-function $\tau(\bt;\e)$. 
We call $\tau(\bt;\e)$ the $\tau$-function of the solution $u(\bt;\e)$ \cite{DYZ18, DZ-norm}.
	
The infinitesimal Galilean symmetry \eqref{KdV-Galilean-Symmetry} for the KdV hierarchy 
can be upgraded to the infinitesimal action~\cite{DZ-norm, KS91} on KdV $\tau$-functions as follows:
\begin{equation}\label{AdditionalSymmetryOnTau}
\frac{\p \tilde\tau(\bt;\epsilon;q)}{\p q}=
\sum_{k\geq 0}t_{k+1}\frac{\p \tilde\tau(\bt;\epsilon;q)}{\p t_{k}}+\frac{t_0^2}{2\e^2}\tilde\tau(\bt;\epsilon;q).
\end{equation} 
\begin{thm}\label{thm-KdV-tau}
If $\tau(\bt;\e)$ is a KdV $\tau$-function, then $\tilde{\tau}(\bt; \e; q)$ defined by 
\begin{equation}\label{GalileanTransformationOnTau}
\tilde{\tau}(\bt; \e; q):=\tau(\bt^{\rm G}(\bt;q);\e)e^{\frac{g(\bt;q)}{\e^2}}
\end{equation}
is a KdV $\tau$-function for arbitrary~$q$, where $\bt^{\rm G}(\bt;q)$ is defined 
in~\eqref{GalileanTransformationOnSolution}, and 
\begin{equation}\label{g-correction}
			g(\bt;q):=\frac{1}{2}\sum_{i,j\geq 0}\frac{q^{i+j+1}}{i+j+1}\frac{t_i}{i!}\frac{t_j}{j!}.
\end{equation}
\end{thm}	
\noindent The proof by a direct verification of the definition of a KdV $\tau$-function is in Section~\ref{S2}. 

\begin{rmk}
The infinitesimal action~\eqref{AdditionalSymmetryOnTau} 
corresponds to an element in $gl(\infty)$~\cite{DJKM, KS91, sato1981soliton}, 
and the corresponding group element sends a point $S$ in the Sato Grassmannian 
satisfying $z^2 S\subset S$ to another point $\tilde{S}$ (depending on the parameter~$q$) satisfying $z^2 \tilde{S} \subset \tilde{S}$. 
In viewpoint of this the claim that $\tilde{\tau}:=\exp(q(\sum_{k\geq 0}t_{k+1}\frac{\p }{\p t_{k}}+\frac{t_0^2}{2\e^2}))(\tau)$ is a KdV $\tau$-function
is direct. 
\end{rmk}
	
The main source for our study 
comes from the growing interest (see~\cite{YZQ23}) in two important partition functions: 
the generalized Br\'ezin--Gross--Witten (BGW) partition function \cite{alexandrov2018, DYZ18, MMS, YZQ21} 
and the non-abelian Born--Infeld (NBI) partition function~\cite{AC, DYZ18}. 
Both of them were known to be particular $\tau$-functions for the KdV hierarchy, and 
some unified formulas for the two were obtained in~\cite{DYZ18} (actually, we will 
also give in Section~\ref{S2} another proof of Theorem~\ref{thm-KdV-tau} based on the method in~\cite{BDY, DYZ18}).
	
Recall that the {\it BGW partition function} $Z^{\rm BGW}(\bt;\e)$ was introduced 
in \cite{BG, GW} (see also~\cite{alexandrov2018, DN, GN92, MMS}), 
and was shown~\cite{DN, GN92, MMS} to be a particular $\tau$-function for the KdV hierarchy. It is often normalized so that 
$Z^{\rm BGW}(\bigzero;\e)\equiv1$.
The logarithm of $Z^{\rm BGW}(\bt;\e)$ has the form
\beq
\log Z^{\rm BGW}(\bt;\e) = \sum_{g,n\ge1} \frac{\e^{2g-2}}{n!}
\sum_{k_1,\dots,k_n\ge0\atop |k|=g-1} c_g(k_1,\dots,k_n) t_{k_1} \dots t_{k_n},
\eeq
where $|k|:=k_1+\dots+k_n$, and $c_g(k_1,\dots,k_n)$ are certain numbers called the {\it BGW numbers}. 
Topological recursion of Chekhov--Eynard--Orantin type 
for this model was given in~\cite{DN}, and explicit $n$-point functions were derived in~\cite{BR, DYZ18}.

In~\cite{MMS} (cf.~\cite{alexandrov2018}), a one-parameter deformation $Z^{\rm gBGW}(N,\bt; \e)$ of 
	the BGW partition 
	function was defined by means of 
	a generalized Kontsevich matrix model. This partition function is often normalized so that $Z^{\rm gBGW}(N,\bigzero; \e)\equiv1$.
	We call $Z^{\rm gBGW}(N,\bt; \e)$ the {\it normalized generalized BGW partition function}.
	Let $x:=N\e\sqrt{-2}$ be the Alexandrov coupling constant. 
	Following~\cite{YZQ21} (cf.~\cite{OS, YZQ23}), 
	define  
	\begin{equation}\label{cBGW-def}
		Z^{\rm cBGW}(x,\bt;\e):=e^{B(x;\e)}Z^{\rm gBGW}\Bigl(\frac{x}{\e\sqrt{-2}},\bt;\e\Bigr),
	\end{equation}
	where
\begin{equation}\label{cBGW-correction}
B(x,\e)=\frac{1}{\e^2}\biggl(\frac{1}{4}x^2\log\Bigl(-\frac{x}{2}\Bigr)-\frac{3}{8}x^2\biggr)+\frac{1}{12}\log\left(-\frac{x}{2}\right)+\sum_{g\geq2}\frac{\e^{2g-2}}{x^{2g-2}}\frac{(-1)^{g}2^{g-1}B_{2g}}{2g(2g-2)} 
\end{equation}
	with $B_j$ being the $j$th Bernoulli number. 
	We call $Z^{\rm cBGW}(x,\bt;\e)$ the {\it corrected generalized BGW partition function}, 
	for short the {\it generalized BGW partition function}.
The logarithm of $Z^{\rm cBGW}(\bt;\e)$ has the form
\beq\label{cBGW-structure}
\log Z^{\rm cBGW}(\bt;\e) = B(x,\e) + \sum_{g,n\ge1} \frac{\e^{2g-2}}{n!}
\sum_{k_1,\dots,k_n\ge0\atop |k|\ge g-1} x^{2|k|-2g+2} c_g(k_1,\dots,k_n) t_{k_1} \dots t_{k_n}.
\eeq
It was shown in~\cite{MMS} (cf.~\cite{alexandrov2018})  that 
	$Z^{\rm cBGW}(x,\bt;\e)$ is a $\tau$-function for the KdV hierarchy for any~$x$.
	In particular, 
	\beq\label{defucbgw}
	u^{\rm cBGW}(x,\bt;\e):=\e^2 \frac{\p^2 \log Z^{\rm cBGW}(x,\bt;\e)}{\p t_0^2}
	\eeq
	is a solution to the KdV hierarchy, which has  
	the initial value~\cite{alexandrov2018, BR, DYZ18, MMS, YZQ21}:
	\begin{equation}\label{cBGWini}
		u^{\rm cBGW}(x,t_0=X,\bigzero;\e)=\frac{\frac{x^2}{4}+\frac{\ee^2}{8}}{(1-X)^2}.
	\end{equation}
	
Before proceeding to the NBI side, we recall some terminologies. For $g,n\ge0$ satisfying $2g-2+n>0$, 
denote by $\overline{\mathcal{M}}_{g,n}$ the moduli space of stable algebraic curves of genus~$g$ with $n$ distinct marked points. Let 
$f: \overline{\mathcal{M}}_{g,n+1}\to \overline{\mathcal{M}}_{g,n}$ be the forgetful map forgetting the last marked point. 
Denote by $\mathcal{L}_i$ ($1\leq i\leq n$) the $i$th cotangent line bundle on $\overline{\mathcal{M}}_{g,n}$ and by $\psi_i$ 
the first Chern class of $\mathcal{L}_i$, and denote $\kappa_j:=f_*(\psi_{n+1}^{j+1})$, $j\ge0$.

In 1990, Witten~\cite{Witten} proposed a conjectural statement on the relationship between $\psi$-class intersection numbers 
and integrable hierarchy, which can be stated as follows. 

\smallskip
	
\noindent {\bf Witten's conjecture.}
{\it The partition function $Z^{\textsc{\tiny\rm WK}}(\bt;\ee)$ of $\psi$-class intersection numbers, defined by 
\begin{equation}\label{WKtau}
	Z^{\rm WK}(\bt;\e):=\exp\biggl(\;\sum_{g,n\geq 0}\tfrac{\e^{2g-2}}{n!}  
		\sum_{k_1,\dots,k_n\geq 0} t_{k_1}\dots t_{k_n} \int_{\overline{\mathcal{M}}_{g,n}} \psi_1^{k_1}\dots\psi_n^{k_n} \;\biggr) \, ,
\end{equation}
is a particular $\tau$-function for the KdV hierarchy, which satisfies the string equation
\begin{equation}\label{WKstring}
\sum_{k\geq 0} t_{k+1}\frac{\p Z^{\textsc{\tiny\rm WK}}(\bt;\e)}{\p t_k}+\frac{t_0^2}{2\e^2}=\frac{\p Z^{\textsc{\tiny\rm WK}}(\bt;\e)}{\p t_1}.
\end{equation}
}

In particular, the function 
$
u^{\rm WK}(\bt;\e):=\e^2 \p^2 \log Z^{\rm WK}(\bt;\e)/\p t_0^2
$
is a solution to the KdV hierarchy, which is characterized by the initial value~\cite{Witten}
\beq\label{WKini23}
u^{\rm WK}(t_0=X,\bigzero;\e)\equiv X.
\eeq
Kontsevich~\cite{Kont} first proved Witten's conjecture. Nowadays, Witten's conjecture 
is also known as the {\it Witten--Kontsevich (WK) theorem}. See~\cite{CLL08, KL07, Mirz, OP} for several other proofs of the WK theorem.
We refer to $u^{\rm WK}(\bt;\e)$ as the {\it WK solution}.
 
Motivated by the work of Ambj{\o}rn--Chekhov~\cite{AC} (cf.~\cite{YZQ23}) we make the following definition.
	\begin{Def}\label{defnbi}
		The NBI partition function $Z^{\rm NBI}(x,\br;\e)$ is defined by 
		\beq\label{NBIDefinition}
		Z^{\rm NBI}(x,\br;\e) := Z^{\rm WK}\bigl(\bt^{\rm WK-NBI}(x,\br);\e\bigr),
		\eeq
		where
		\beq 
		t^{\rm WK-NBI}_i(x,\br):=r_i + (2i-1)!! \Bigl(-\frac2{x^2}\Bigr)^i - \delta_{i,0} + \delta_{i,1}, \quad i\ge0.
		\eeq
	\end{Def}
	
It is easy to check that $Z^{\rm NBI}(x,\br;\e)$ is a well-defined power series of~$\br$. 
The celebrated results in~\cite{KMZ, LX, MZ} (see also~\cite{Faber, MS}) lead to the following proposition.
\begin{pro} \label{NBIkappaprop}
We have
\begin{align}
Z^{\rm NBI}(x,\br;\e) = \bigl(\tfrac{x^2}{2}\bigr)^{\frac1{24}}\exp\biggl(\,\sum_{g,n\geq 0}\tfrac{\e^{2g-2}}{n!}
		\sum_{k_1,\dots,k_n\geq 0\atop |k|\leq 3g-3+n} \bigl(\tfrac{x^2}2\bigr)^{|k|+1-g} 
		\int_{\overline{\mathcal{M}}_{g,n}}
		K_{3g-3+n-|k|} \prod_{i=1}^{n}\psi_{i}^{k_i} r_{k_i}\,\biggr).
\end{align}
Here,  $K_m\in H^{2m}(\overline{\mathcal{M}}_{g,n};\mathbb{C})$, $m\ge0$, are polynomials of $\kappa$-classes defined via 
\beq
\sum_{m\ge0} K_m :=e^{\sum_{j\geq 1}s_{j}^{\rm NBI}\kappa_{j}}=:\mathbb{K},
\eeq
where $\bigs^{\rm NBI}=(s_1^{\rm NBI}, s_2^{\rm NBI}, s_3^{\rm NBI}, \dots)$ 
		is a sequence of rational numbers defined by 
		\beq\label{NBI-Schur}
		e^{-\sum_{j\ge1} s^{\rm NBI}_j z^j} = \sum_{k\ge0} (-1)^k (2k+1)!! z^k.
		\eeq
\end{pro}
\noindent Details of the proof of this proposition are given in Section~\ref{S3}. 

It turns out that the class $\mathbb{K}$ defined above  
coincides with the one introduced by Kazarian--Norbury in~\cite{KN}.	
We refer to $\mathbb{K}$ or say $K_m$ as the {\it Kazarian--Norbury--NBI class}. 
We observe that the first few renormalized numbers 
	$(-1)^j j s_j^{\rm NBI}/3$ are  $$1, ~ 7, ~ 69, ~ 843, ~ 12081;$$ see e.g.
the On-Line Encyclopedia of Integer Sequences\footnote{https://oeis.org} (OEIS) No. A226270.
	
It follows from Definition~\ref{defnbi} and the WK theorem that 
the NBI partition function $Z^{\rm NBI}(x,\br;\e)$ is a particular $\tau$-function for the KdV hierarchy. So the function 
\beq
	u^{\rm NBI}(x,\br;\epsilon):=\e^2 \frac{\p^2 \log Z^{\rm NBI}(x,\br;\e)}{\p r_0^2}
\eeq
is a particular solution to the KdV hierarchy, which we refer to as the {\it NBI solution}.
The following lemma characterizes the NBI solution, whose proof is in Section~\ref{S3}.

\begin{lem}\label{nbikdv}
		The initial value for the NBI solution 
		is given by  
		\begin{equation}\label{NBIini}
			u^{\rm NBI}(x,r_{0}=X,\bigzero;\epsilon)=\frac{\frac{x^2}{4}+\frac{\e^2}{8}}{(1-X)^2}-\frac{x^2}{4}.
		\end{equation}
\end{lem}

	Observing that the initial values~\eqref{NBIini}, \eqref{cBGWini} 
	differ only by the constant 
	$x^2/4$, we prove in Section~\ref{S4} 
	the following theorem, which we refer to as the {\it BGW--NBI correspondence}.
	
\begin{thm}\label{thm-cBGW-NBI}
	The following two identities hold:
\begin{align}
&Z^{\rm cBGW}(x,\bt;\e) 
= Z^{\rm NBI}\bigl(x,\bt^{\rm G}\bigl(\bt;\tfrac{x^2}{4}\bigr);\epsilon\bigr) e^{-\frac{\log2}{24}-\frac{3x^2}{8\e^2}+\frac{x^2}{4\e^2}\log(-\frac{x}{2})
	+\frac{C_1(x,\bt)}{\e^2}+\frac{C_2(x,\bt)}{\e^2}}, \label{cBGW-NBI}\\
& Z^{\rm NBI}(x,\br;\e) 
= Z^{\rm cBGW}\bigl(x,\bt^{\rm G}\bigl(\br;-\tfrac{x^2}{4}\bigr); \epsilon\bigr) e^{\frac{\log2}{24}+\frac{3x^2}{8\e^2}-\frac{x^2}{4\e^2}\log(-\frac{x}{2})
	+\frac{M_1(x,\br)}{\e^2}+\frac{M_2(x,\br)}{\e^2}},\label{NBI-cBGW}
\end{align}
where $\bt^{\rm G}(\bt;q)$ is defined in~\eqref{GalileanTransformationOnSolution}, and
\beq\label{C1C2}
C_1(x,\bt) =\sum_{i\geq0}\frac{(\frac{x^2}{4})^{i+1}t_i}{(i+1)!(2i+1)},\quad
C_2(x,\bt) = \frac{1}{2}\sum_{i,j\ge 0}\frac{(\frac{x^2}{4})^{i+j+1}t_i t_j}{i!j!(i+j+1)},
\eeq 	
\beq
M_1(x,\br)= \sum_{i\geq0}\bigl(\tfrac{2^{i+1}(i+1)!}{(2i+1)!!}-1\bigr)\frac{(-\frac{x^2}{4})^{i+1}r_i}{(i+1)!}, \quad 
M_2(x,\br)= \frac{1}{2}\sum_{i,j\ge 0}\frac{(-\frac{x^2}{4})^{i+j+1}r_ir_j}{i!j!(i+j+1)}.
\eeq 
\end{thm} 

\smallskip
		
The paper is organized as follows. In Section~\ref{S2}, we prove Theorems \ref{thm-of-kdv-solution}, \ref{thm-KdV-tau}. 
In Section~\ref{S3}, 
we study on the NBI partition function. 
In Section~\ref{S4}, we give applications of Theorem~\ref{thm-KdV-tau}. Generalizations to Frobenius manifolds 
are given in Section~\ref{S5}.

\section{The Galilean symmetry of the KdV hierarchy}\label{S2}
	In this section, we prove Proposition~\ref{prop-galilean-solution} and Theorem~\ref{thm-KdV-tau}. 
	
	Before entering into the proofs, let us give two ways of setting up the rigorous aspects for $\tau$-functions:
	\begin{itemize}
	\item[(i)] We can understand that a KdV $\tau$-function $\tau(\bt;\e)$ has the form
	\beq
	\tau(\bt;\e)= \hat \tau(\bt;\e) e^{\e^{-2}\mathcal{F}_0(\bt)},
	\eeq
	where 
	\beq
	\hat \tau(\bt;\e) \in \CC[[\e^2]] [[\bt]], \quad \mathcal{F}_0(\bt)\in \CC[[\bt]].
	\eeq
	\item[(ii)] We can understand that a KdV $\tau$-function $\tau(\bt;\e)$ has the form
	\beq
	\tau(\bt;\e)= \tau_{\rm norm}(\bt;\e) e^{B(\e)},
	\eeq
	where 
	\beq
	\tau_{\rm norm}(\bt;\e) \in \CC((\e^2)) [[\bt]], \quad \tau_{\rm norm}(\bigzero;\e) \equiv 1,  \quad \e^2 B(\e)\in \CC[[\e^2]].
	\eeq
	\end{itemize}
	It can be checked all the partition functions that appear in this paper 
	are well defined from both of the above ways, and we will not say more details about this.
	
	Let us now prove Proposition~\ref{prop-galilean-solution}.
	
	\begin{proof}[Proof of Proposition~\ref{prop-galilean-solution}]
	It is easy to verify that 
	$\tilde{u}(\bt;\e;q)$ defined by \eqref{GalileanTransformationOnSolution} 
	is an element of $\mathbb{C}[[\bt,\e,q]]$. 
	Let us show that $\tilde{u}(\bt;\e;q)$ satisfies \eqref{IVP-for-solution}. Indeed, 
	noticing that
		\beq \label{qderitime}
		\frac{\p t_{i}^{\rm G}(\bt;q)}{\p q}=t_{i+1}^{\rm G}(\bt;q),\quad i\geq 0,
		\eeq
		we find
		\begin{align}
			\mbox{LHS of \eqref{IVP-for-solution}}=
			\sum_{i\geq0}t^{\rm G}_{i+1}(\bt;q)\frac{\p u}{\p t_{i}}\bigl(\bt^{\rm G}(\bt;q);\e\bigr) +1.
		\end{align}
		We also have 
		\begin{align}
			\mbox{RHS of \eqref{IVP-for-solution}}=
			&
			\sum_{k\geq0}\sum_{i\geq 0}t_{k+1}\frac{\p u}{\p t_{i}}\bigl(\bt^{\rm G}(\bt;q);\e\bigr)\frac{\p t^{\rm G}_i(\bt;q)}{\p t_k}+1\notag\\
			=&\sum_{i\geq 0}\frac{\p u}{\p t_{i}}\bigl(\bt^{\rm G}(\bt;q);\e\bigr)\sum_{k\geq i}\frac{q^{k-i}}{(k-i)!}t_{k+1}+1\notag\\
			=&\sum_{i\geq0}t^{\rm G}_{i+1}(\bt;q)\frac{\p u}{\p t_{i}}\bigl(\bt^{\rm G}(\bt;q);\e\bigr) +1,
		\end{align}
		which equals the left-hand side of~\eqref{IVP-for-solution}. 
	The uniqueness is guaranteed by a recursion of the coefficients of powers of $q$.
	\end{proof}
	
	Similar to Proposition~\ref{prop-galilean-solution}, we have the following proposition.
\begin{pro}\label{pro-tau-Galilean-symmetry}
For any given function $\tau(\bt;\e)$, the initial value problem of equation~\eqref{AdditionalSymmetryOnTau}
with the initial data
\begin{align}
		\tilde{\tau}(\bt;\e;q=0)=\tau(\bt;\e)\label{IVP-tau}
\end{align}
has a unique power-series in~$q$ solution $\tilde{\tau}(\bt;\e;q)$, 
which has the explicit expression~\eqref{GalileanTransformationOnTau}. 
\end{pro}
	
	\begin{proof}
		It is easy to show that 
		$\tilde{\tau}(\bt;\e;q)$ defined by~\eqref{GalileanTransformationOnTau} is indeed a power series of~$q$.
		 Let us verify that $\tilde{\tau}(\bt;\e;q)$ satisfies \eqref{IVP-tau}. On one hand,
		\begin{align}
			\mbox{LHS of  \eqref{AdditionalSymmetryOnTau}}
			=\Biggl(\sum_{i\geq 0}t^{\rm G}_{i+1}(\bt;q)
			\frac{\p \tau(\bt^{\rm G}(\bt;q);\epsilon)}{\p t_{i}}
			+\frac{\tau(\bt^{\rm G}(\bt;q);\epsilon)}{2\e^2}
			\sum_{i,j\geq 0}q^{i+j}\frac{t_it_j}{i!j!}\Biggr)e^{\frac{g(q,\bt)}{\e^2}}.
		\end{align}
	On the other hand,
		\begin{align}
			&\mbox{RHS of \eqref{AdditionalSymmetryOnTau}}\notag\\
			&=\Biggl(\sum_{k\geq 0}t_{k+1}\sum_{i\geq0}\frac{\p \tau}{\p t_{i}}\bigl(\bt^{\rm G}(\bt;q);\epsilon\bigr)\frac{\p t^{\rm G}_{i}(\bt;q)}{\p t_{k}}+\frac{1}{\e^2}\Biggl(\sum_{i,j\geq0}\frac{q^{i+j+1}}{i+j+1}\frac{t_{i} t_{j+1}}{i!j!}+\frac{t_0^2}{2}\Biggr)\tau\bigl(\bt^{\rm G}(\bt;q);\epsilon\bigr)\Biggr)e^{\frac{g(\bt;q)}{\e^2}}\notag\\
			&=\Biggl(\sum_{i\geq 0}\frac{\p \tau}{\p t_{i}}\bigl(\bt^{\rm G}(\bt;q);\epsilon\bigr)\sum_{k\geq i}t_{k+1}\frac{q^{k-i}}{(k-i)!}\notag\\
			&+\frac{1}{2\e^2}\Biggl(\sum_{i\geq 0,j\geq 1}\frac{q^{i+j}}{i+j}\frac{t_{i}t_{j}}{i!(j-1)!}+\sum_{i\geq 1,j\geq 0}\frac{q^{i+j}}{i+j}\frac{t_{i}t_{j}}{(i-1)!j!}+t_0^2\Biggr)\tau\bigl(\bt^{\rm G}(\bt;q);\epsilon\bigr)\Biggr)e^{\frac{g(\bt;q)}{\e^2}}\notag\\
			&= \Biggl(\sum_{i\geq 0}\frac{\p \tau(\bt^{\rm G}(\bt;q);\epsilon)}{\p t_{i}}t^{\rm G}_{i+1}(\bt;q)
			+\frac{\tau(\bt^{\rm G}(\bt;q);\epsilon)}{2\e^2}\sum_{i,j\geq 0}q^{i+j}\frac{t_it_j}{i!j!}\Biggr)e^{\frac{g(\bt;q)}{\e^2}},
		\end{align}
which finishes the verification. The uniqueness is guaranteed by a recursion of the coefficients of powers of~$q$.
	\end{proof}
	
	\begin{rmk}
	We note that Proposition~\ref{pro-tau-Galilean-symmetry} can also be easily obtained from the Givental quantization formula~\cite{Gi01}.
	\end{rmk}	

	We are ready to prove Theorem~\ref{thm-KdV-tau}.
	\begin{proof}[Proof of Theorem~\ref{thm-KdV-tau}]
		The goal is to show that for arbitrary $\bt,\bt'$ and for all $p\ge 0$,
		\begin{small}
			\begin{align}\label{hibispe}
				\underset{z=\infty}{\rm res} \,\tau\bigl(\bt^{\rm G}(\bt-[z^{-1}];q);\e\bigr)
				\,\tau\bigl(\bt^{\rm G}(\bt'+[z^{-1}];q);\e\bigr)
				e^{\frac{g({\bt}-[z^{-1}];q)+g(\bt'+[z^{-1}];q)}{\e^2}}e^{\sum\limits_{j\geq 0}\frac{t_{j}-t'_{j}}{\e}\frac{z^{2j+1}}{(2j+1)!!}}z^{2p}dz=0.
			\end{align}
		\end{small}

		By a direct calculation, we have 
		\begin{flalign}\label{calc1} 
			t_{n}^{\rm G}\bigl(\bt-[z^{-1}];q\bigr)
	=t^{\rm G}_{n}(\bt;q)-\e\frac{(2n-1)!!}{w(z,q)^{2n+1}},\quad n\ge 0,
		\end{flalign}
		where $w(z,q):=\sqrt{z^2-2q}$. Using the Newton--Leibniz formula,  \eqref{g-correction}, and \eqref{GalileanTransformationOnSolution}, we have
		\begin{flalign}\label{calc2pre} 
			&
			 g({\bt}-[z^{-1}];q)=\int_{0}^{q}\frac{\partial g({\bt}-[z^{-1}];s)}{\partial s}ds
			=\frac{1}{2}\int_{0}^{q}\biggl(\sum_{i\geq 0}\frac{s^i}{i!}t_{i}-\e\sum_{i\geq0}\frac{s^i}{i!}\frac{(2i-1)!!}{z^{2i+1}}\biggr)^2 ds\notag\\
			&=\frac{1}{2}\int_{0}^{q}\biggl(t^{\rm G}_{0}(\bt;s)-\frac{\e}{w(z,s)}\biggr)^2 ds\notag\\
			&=\frac{1}{2}\int_{0}^{q} \bigl(t^{\rm G}_{0}(\bt;s)\bigr)^2 ds- \e\int_{0}^{q}\frac{t^{\rm G}_{0}(\bt;s)}{w(z,s)}ds
			+\frac{\e^2}{2}\int_{0}^{q}\frac{1}{w(z,s)^2} ds.
		\end{flalign}
		Observing that 
		\beq\label{wdq}
		\frac{\partial (w(z,q)^n)}{\partial q}=-n w(z,q)^{n-2}, \quad n\geq 1,
		\eeq
		we arrive at
		\beq\label{calc2} 
			g({\bt}-[z^{-1}];q)=
			\frac{1}{2}\int_{0}^{q} \bigl(t^{\rm G}_{0}(\bt;q)\bigr)^2 ds 
			- \e\int_{0}^{q}\frac{t^{\rm G}_{0}(\bt;s)}{w(z,s)}ds
			+\frac{\e^2}{2}\log\biggl(\frac{z}{w(z,q)}\biggr).
		\eeq	
		Using \eqref{GalileanTransformationOnSolution}, the Newton--Leibniz formula, and
		 \eqref{qderitime}, we obtain
\begin{flalign}\label{calc3} 
&\sum_{j\geq 0}\frac{t_{j}}{\e}\frac{z^{2j+1}}{(2j+1)!!}\notag\\
&=\sum_{j\geq 0}\frac{t^{\rm G}_{j}(\bt;q)}{\e}\frac{w(z,q)^{2j+1}}{(2j+1)!!}
			-\int_{0}^{q}\frac{\partial}{\partial s}\Biggl(\sum_{j\geq 0}\frac{t^{\rm G}_{j}(\bt;s)}{\e}\frac{w(z,s)^{2j+1}}{(2j+1)!!}\Biggr)ds\notag\\
&=\sum_{j\geq 0}\frac{t^{\rm G}_{j}(\bt;q)}{\e} \frac{w(z,q)^{2j+1}}{{(2j+1)!!}}+\frac1{\e}\int_{0}^{q}\frac{t^{\rm G}_{0}(\bt;s)}{w(z,s)}ds.&
\end{flalign}
		Also,
\begin{flalign}\label{calc4}
	z^{2p} dz
=\sum_{l=0}^{p}\binom{p}{l}(2q)^{p-l}\frac{w(z,q)^{2l+1}}{z}dw.
\end{flalign}
		Similarly, we have
		\beq
		\label{calc5} 
			t_{n}^{\rm G}(\bt'+[z^{-1}];q)
			=t^{\rm G}_{n}(\bt';q)+\e\frac{(2n-1)!!}{w(z,q)^{2n+1}},\quad n\ge 0,
		\eeq
\beq\label{calc6} 
g({\bt'}-[z^{-1}];q)=\frac{1}{2}\int_{0}^{q}t^{\rm G}_{0}(\bt';q)^2 ds+\e\int_{0}^{q}\frac{t^{\rm G}_{0}(\bt';s)}{w(z,s)}ds+\frac{\e^2}{2}\log\biggl(\frac{z}{w(z,q)}\biggr).
\eeq
		
	Substituting \eqref{calc1}, \eqref{calc2}, \eqref{calc3}, \eqref{calc4}, \eqref{calc5} and \eqref{calc6} in~\eqref{hibispe}, we find
			\begin{equation} \begin{split} 
					&\mbox{LHS of \eqref{hibispe}}\notag\\
					&=e^{\frac{1}{2\epsilon^2}\int_{0}^{q} (t^{\rm G}_{0}(\bt;s))^2+(t^{\rm G}_{0}(\bt';s))^2 ds}
					\sum_{l=0}^{p}\binom{p}{l}(2q)^{p-l}\\
					&\times\underset{w=\infty}{\rm res}\,\tau\bigl(\bt^{\rm G}(\bt;q)-[w^{-1}];\e\bigr)\,\tau\bigl(\bt^{\rm G}(\bt';q)+[w^{-1}];\e\bigr) 
					e^{\sum_{j\geq 0}\frac{t^{\rm G}_{j}(\bt;q)-t^{\rm G}_{j}(\bt';q)}{\e}\frac{w^{2j+1}}{(2j+1)!!}}w^{2l} dw,
			\end{split}	\end{equation}
		which vanishes because $\tau(\bt;\e)$ is a KdV $\tau$-function. The theorem is proved.
	\end{proof}
	
	As a corollary of Theorem~\ref{thm-KdV-tau}, let us prove Theorem~\ref{thm-of-kdv-solution}.
	\begin{proof}[Proof of Theorem~\ref{thm-of-kdv-solution}]
		Take $\tau(\bt;\e)$ a specific $\tau$-function of the solution $u(\bt;\e)$ to the KdV hierarchy.
		According to Theorem~\ref{thm-KdV-tau}, we know that $\tilde{\tau}(\bt;\e;q)$ defined by~\eqref{GalileanTransformationOnTau}
		 is also a KdV $\tau$-function. Noticing that $\p t^{\rm G}_{k}(\bt;q)/\p t_0=\delta_{k,0}$, $k\geq 0$, we find
	\begin{equation}
			\e^2\frac{\p^2\log(\tilde{\tau}(\bt;\e,q))}{\p t_{0}^2}
			=u(\bt^{\rm G}(\bt;q);\e)+q=\tilde{u}(\bt;\e;q).
		\end{equation}
	Thus $\tilde{u}(\bt;\e,q)$ is a solution to the KdV hierarchy.
	\end{proof}
	
For an arbitrary solution $u(\bt; \e)$ to the KdV hierarchy, let $\tau(\bt; \e)$ be the 
	$\tau$-function of the solution $u(\bt; \e)$. Recall that the $n$th logarithmic derivatives 
	\beq
	\frac{\p^{n} \log\tau}{\p t_{k_1}\dots \p t_{k_n}}(\bt;\e), \quad k_1,\dots,k_n\ge0,
	\eeq
	are called {\it $n$-point correlation functions associated with~$\tau$}. Their restrictions 
	\beq
	\frac{\p^{n} \log\tau}{\p t_{k_1}\dots \p t_{k_n}}(X, t_1=0, t_2=0, \dots; \e), \quad k_1,\dots,k_n\ge0,
	\eeq
	are called {\it $n$-point partial correlation functions associated with~$\tau$}. 
	For $n\ge1$, define  
	\begin{align}
	&W_{n}(z_1,\dots,z_n;X;\e)
	=\e^n \sum_{k_1,\dots,k_n\ge0} \frac{\p^{n} \log\tau}{\p t_{k_1}\dots \p t_{k_n}}
	(X,\bigzero;\e) 
	\frac{(2k_1+1)!!\dots(2k_n+1)!!}{z_1^{2k_1+3} \dots z_n^{2k_n+3}} ,\\
	&C_{n}(x_1,\dots,x_n;X;\e)
	=\e^n \sum_{k_1,\dots,k_n\ge0} \frac{\p^{n} \log\tau}{\p t_{k_1}\dots \p t_{k_n}}(X,\bigzero;\e) x_1^{k_1} \dots x_n^{k_n},
	\end{align}
which are known as two different types of $n$-point functions~\cite{BDY} associated with~$\tau$. 
	
	From the definition~\eqref{GalileanTransformationOnTau} of~$\tilde{\tau}$, we know that for $n\ge 1$ and $k_1,\dots,k_n \ge 0$,
	\begin{align}
		\frac{\p^{n} \log\tilde{\tau}}{\p t_{k_1}\dots \p t_{k_n}}(\bt;\e;q)
		&=\sum_{0\le j_1\le k_1,\dots,0\le j_n\le k_n}\frac{q^{|k|-|j|}}{\prod_{i=1}^{n}(k_i-j_i)!}\frac{\p^n \log \tau}{\p t_{j_1}\dots \p t_{j_n}}(\bt^{\rm G}(\bt;q);\e)
		\notag\\
		&+\frac{\delta_{n,1}}{\e^2}\sum_{j\ge 0}\frac{q^{k_1+j+1}t_j}{(k_1+j+1)k_1!j!}+\frac{\delta_{n,2}}{\e^2}\frac{q^{k_1+k_2+1}}{(k_1+k_2+1)k_1!k_2!}.
	\end{align}
Taking $t_1=t_2=\dots=0$, we find, for $n\ge 1$ and $k_1,\dots,k_n \ge 0$, that 
	\begin{align}\label{Galilean-Correlator}
		\frac{\p^{n} \log\tilde{\tau}}{\p t_{k_1}\dots \p t_{k_n}}(X,\bigzero;\e;q)
		&=\sum_{0\le j_1\le k_1,\dots,0\le j_n\le k_n}\frac{q^{|k|-|j|}}{\prod_{i=1}^{n}(k_i-j_i)!}\frac{\p^n \log \tau}{\p t_{j_1}\dots \p t_{j_n}}(X, \bigzero;\e)
		\notag\\
		&+\frac{\delta_{n,1}}{\e^2}\frac{q^{k_1+1}X}{(k_1+1)!}+\frac{\delta_{n,2}}{\e^2}\frac{q^{k_1+k_2+1}}{(k_1+k_2+1)k_1!k_2!}.
	\end{align}
Denote by $\widetilde{W}_{n}(z_1,\dots,z_n;X;\e;q)$, $\widetilde{C}_{n}(x_1,\dots,x_n;X;\e;q)$ the two different types of $n$-point functions 
associated with~$\tilde{\tau}$. Then we have for $n\ge1$,
	\begin{align}\label{WWt}
			\widetilde{W}_{n}(z_1,\dots,z_n;X;\e;q)=\,&W_{n} \bigl( {\textstyle \sqrt{z_1^2-2q}},\dots,\sqrt{z_n^2-2q};X;\e\bigr)
			+\tfrac{\delta_{n,1}}{\e}\Bigl(\tfrac{1}{\sqrt{z_1^2-2q}}-\tfrac{1}{z_1}\Bigr)X\notag\\
			&+\tfrac{\delta_{n,2}}{ (z_1^2-z_2^2)^2}
			\Bigl(\tfrac{\sqrt{z_1^2-2q}}{\sqrt{z_2^2-2q}}+\tfrac{\sqrt{z_2^2-2q}}{\sqrt{z_1^2-2q}}
			-\tfrac{z_1}{z_2}-\tfrac{z_2}{z_1}\Bigr),
	\end{align}

\begin{align}\label{CCt}
		\widetilde{C}_{n}(x_1,\dots,x_n;X;\e;q)&=e^{q\sum\limits_{i=1}^{n}x_i}C_{n}\left(x_1,\dots,x_n; X;\e\right)
		+\tfrac{\delta_{n,1}}{\e}\tfrac{e^{q x_1}-1}{x_1}X
		+\delta_{n,2}\tfrac{e^{q (x_1+x_2)}-1}{x_1+x_2}.
\end{align}

We observe that Theorem~\ref{thm-KdV-tau} is equivalent to the validity of the above identity~\eqref{WWt} (or~\eqref{CCt}) for $n\ge2$.
Based on this and on the matrix-resolvent method \cite{BDY, DYZ18}, let us give an alternative proof of Theorem~\ref{thm-KdV-tau}.
\begin{proof}[Another proof of Theorem~\ref{thm-KdV-tau}.]
Let $u(\bt;\e)=\e^2 \p_X^2(\log \tau(\bt;\e))$. Since $\tau(\bt;\e)$ is a KdV $\tau$-function, we know that $u(\bt;\e)$ 
is a solution to the KdV hierarchy which has the initial value $u(X, \bigzero; \e)=:f(X;\e)$. 
Let $R(\lambda,X;\e)$ be the basic matrix resolvent of~$u(\bt;\e)$ evaluated at $t_1=t_2=\dots=0$ 
(for the precise definition for $R(\lambda,X;\e)$ see~\cite{DYZ18}), and let 
$\widehat{R}(\lambda,X;\e;q)$ be the basic matrix resolvent of~$\hat u(\bt;\e;q)$ evaluated at $t_1=t_2=\dots=0$, 
where $\hat u(\bt;\e;q)$ is the unique solution the KdV hierarchy with the initial value $f(X;\e)+q$.
Take $\hat \tau(\bt;\e;q)$ a $\tau$-function of $\hat u(\bt;\e;q)$ and let $\widehat{W}_n(z_1,\dots,z_n;X;\e;q)$ be $n$-point functions 
associated with~$\hat{\tau}(\bt;\e;q)$.
It is straightforward to verify that $\widehat{R}(z^2,X;\e;q)=\frac{z}{\sqrt{z^2-2q}} R(z^2-2q, X;\e)$, from which 
we find that \eqref{WWt} holds for $n\ge2$ with $\widetilde{W}_n$ on the left-hand side being replaced by $\widehat{W}_n$.
The theorem is proved.
\end{proof}

\section{The NBI partition function}\label{S3}
In this section, we give a detailed study on the NBI partition function. 
	
Recall that Witten's conjecture can be equivalently~\cite{DVV, KS91} 
	stated as follows: $Z^{\rm WK}=Z^{\rm WK}(\bt;\e)$ satisfies 
	the following infinite set of linear equations 
	\begin{equation}\label{WK-Virasoro}
		L^{\rm WK}_{k} \bigl(Z^{\rm WK}\bigr)=0, \quad k\geq -1,
	\end{equation}
	where 
	\begin{align}
		L^{\rm WK}_{k}:=&\frac{\e^2}{2}\sum_{i+j=k-1}\frac{(2i+1)!!(2j+1)!!}{2^{k+1}}\frac{\p^2}{\p t_{i}\p t_{j}}\notag
		\\
		+&\sum_{i\geq 0}\frac{(2k+2i+1)!!}{2^{k+1}(2i-1)!!}(t_i-\delta_{i,1})\frac{\p}{\p t_{i+k}}+\frac{\delta_{k,0}}{16}
		+\frac{t_0^2}{2\e^2}\delta_{k,-1}.
	\end{align}
	These equations are called \emph{Virasoro constraints} for~$Z^{\rm WK}$. Clearly, the $k=-1$ equation in~\eqref{WK-Virasoro} 
	is nothing but the string equation~\eqref{WKstring}.
The Witten--Kontsevich partition function also satisfies the following \emph{dilaton equation}:
\beq\label{WKdilaton}
\sum_{k\geq 0} t_k\frac{\p Z^{\rm WK}(\bt;\e)}{\p t_k}+\e\frac{\p Z^{\rm WK}(\bt;\e)}{\p \e}+\frac{1}{24}=\frac{\p Z^{\rm WK}(\bt;\e)}{\p t_1}.
\eeq
		
Denote by $Z^{\kappa}$ the partition function of integrals of mixed $\psi$-, $\kappa$- classes on $\overline{\mathcal{M}}_{g,n}$
	\begin{equation}\label{WPDefinition}
		Z^{\kappa}(\bt;\e; \bigs):=\exp\left(\sum_{g,n\ge0}\frac{\e^{2g-2}}{n!}\sum_{k_1,\dots,k_n\geq0}
		\int_{\overline{\mathcal{M}}_{g,n}}e^{\sum_{j\geq 1}\kappa_j s_j} \prod_{i=1}^{n}\psi_{i}^{k_i} t_{k_i}\right),
	\end{equation}
	where $\bigs=(s_1,s_2,s_3,\cdots)$. 
We note that the integrals 
\beq
\int_{\overline{\mathcal{M}}_{g,n}} \psi_1^{k_1}\dots\psi_n^{k_n} \kappa_{j_1} \dots \kappa_{j_s}, \quad k_1,\dots,k_n\ge0, \, j_1,\dots,j_s\ge1,
\eeq
vanish unless the following degree-dimension matching holds:
\beq\label{dd}
k_1+\dots+k_n+j_1+\dots+j_s=3g-3+n.
\eeq
	It was proved in \cite{KMZ, LX, MZ} that 
\begin{equation}\label{WK-Shift-WP}
		Z^{\kappa}(\bt;\e;\bigs)=Z^{\rm WK}(t_0,t_1,t_2-p_1(\bigs),t_{3}-p_2(\bigs),\dots;\e)
	\end{equation}
	where $p_j(\bigs)$, $j\ge1$, are polynomials of~$\bigs$ defined via
	\begin{equation}\label{WPWK}
		e^{-\sum_{i\geq 1}s_{i}z^{i}}=\sum_{j\geq 0}p_{j}(\bigs)z^{j},\quad p_0(\bigs):=1.
	\end{equation}
	
We are ready to prove Proposition~\ref{NBIkappaprop}.
\begin{proof}[Proof of Proposition~\ref{NBIkappaprop}]
Putting $z=\frac{2}{x^2}w$ in \eqref{NBI-Schur}, we get
\beq
\exp\biggl(-\sum_{i\geq 1}s_{j}^{\rm NBI}\biggl(\frac{2}{x^2}\biggr)^{j}w^{j}\biggr)=\sum_{k\geq 0}\biggl(-\frac{2}{x^2}\biggr)^{k}(2k+1)!!w^k.
\eeq
Then, by using~\eqref{WK-Shift-WP}, we find
\begin{align}
&Z^{\rm WK}(\bt;\e)|_{t_i=b_i-(2i-1)!!(-\frac{2}{x^2})^{i-1}-\frac{x^2}{2}\delta_{i,0}+\delta_{i,1}, \,i\ge0} \notag\\
&=\exp\Biggl(\sum_{g,n\ge0}\frac{\e^{2g-2}}{n!}\sum_{k_1,\dots,k_n\geq0}
\int_{\overline{\mathcal{M}}_{g,n}}\prod_{i=1}^{n}\psi_{i}^{k_i}b_{k_i}
\exp\biggl(-\sum_{j\geq1} s_{j}^{\rm NBI} \biggl(\frac{2}{x^2}\biggr)^{j}\kappa_{j}\biggr)\Biggr)\notag\\
&=\exp\Biggl(\sum_{g,n\ge 0}\frac{\e^{2g-2}}{n!}\sum_{k_1,\cdots,k_n\ge 0}\biggl(\frac{x^2}{2}\biggr)^{|k|-n-3g+3}\int_{\overline{\mathcal{M}}_{g,n}}\prod_{i=1}^{n}\psi_{i}^{k_i}b_{k_i}\mathbb{K}\Biggr),
\end{align}
where in the second equality we have used the degree-dimension matching~\eqref{dd}.
So by replacing $b_i$ by $\frac{x^2}{2}r_i$, $i\ge0$, and $\e$ by $\frac{x^2}{2}\e$, we obtain 
\begin{align}
&Z^{\rm WK}\bigl(\bt; \tfrac{x^2}{2}\e\bigr)
\big|_{t_i=\frac{x^2}{2}r_i-(2i-1)!!(-\frac{2}{x^2})^{i-1}-\frac{x^2}{2}\delta_{i,0}+\delta_{i,1}, \,i\ge0} \notag\\
&=\exp\Biggl(\sum_{g,n\ge 0}\tfrac{\e^{2g-2}}{n!}\sum_{k_1,\dots,k_n\ge 0}\bigl(\tfrac{x^2}{2}\bigr)^{|k|-g+1}\int_{\overline{\mathcal{M}}_{g,n}}\prod_{i=1}^{n}\psi_{i}^{k_i}r_{k_i}\mathbb{K}\Biggr).
\end{align}
Now by using~\eqref{WKdilaton}, we have
\begin{align}
& 
Z^{\rm WK}\bigl(\bt; \tfrac{x^2}{2}\e\bigr)
\big|_{t_i=\frac{x^2}{2}r_i-(2i-1)!!(-\frac{2}{x^2})^{i-1}-\frac{x^2}{2}\delta_{i,0}+\delta_{i,1}, \,i\ge0} \notag\\
&= e^{\frac1{24} \log \left(\frac2{x^2}\right)} 
Z^{\rm WK}\left(\bt; \e\right)\big|_{t_i=r_i+(2i-1)!!(-\frac{2}{x^2})^{i}-\delta_{i,0}+\delta_{i,1}, \,i\ge0} .
\end{align}
The proposition is proved.
\end{proof}

Let 
\beq
\mathcal{F}^{\rm WK}(\bt;\e):=\log Z^{\rm WK}(\bt;\e),  
\eeq
be the {\it WK free energy}. 
By definition it has the genus expansion
\beq
\mathcal{F}^{\rm WK}(\bt;\e)=: \sum_{g\ge0} \e^{2g-2} \mathcal{F}^{\rm WK}_g(\bt),
\eeq
where $\mathcal{F}^{\rm WK}_g(\bt)$ is called the {\it genus~$g$ part of the WK free energy} (for short the 
{\it genus~$g$ WK free energy}). 
Similarly, let 
\beq
\mathcal{F}^{\rm NBI}(x,\br;\e):=\log Z^{\rm NBI}(x,\br;\e), \quad 
\mathcal{F}^{\rm cBGW}(x,\bt;\e):=\log Z^{\rm cBGW}(x,\bt;\e)
\eeq
be the {\it NBI, cBGW free energies}, respectively. 
We have
\beq
\mathcal{F}^{\rm NBI}(x,\br;\e)=: \sum_{g\ge0} \e^{2g-2} \mathcal{F}^{\rm NBI}_g(x,\br), \quad 
\mathcal{F}^{\rm cBGW}(x,\bt;\e)=: \sum_{g\ge0} \e^{2g-2} \mathcal{F}^{\rm cBGW}_g(x,\bt),
\eeq
where $\mathcal{F}^{\rm NBI}_g(x,\br)$ is called  the 
{\it genus~$g$ NBI free energy} and  $\mathcal{F}^{\rm cBGW}_g(x,\br)$ is called  the 
{\it genus~$g$ generalized BGW free energy}. 

Let 
\beq
v^{\rm WK}(\bt):= \frac{\p^2 \mathcal{F}^{\rm WK}_0(\bt)}{\p t_0^2}, \quad 
v^{\rm NBI}(x,\br):= \frac{\p^2 \mathcal{F}^{\rm NBI}_0(x,\br)}{\p r_0^2}, \quad 
v^{\rm cBGW}(x,\bt):= \frac{\p^2 \mathcal{F}^{\rm cBGW}_0(x,\bt)}{\p t_0^2}.
\eeq
It is well known that $v^{\rm WK}(\bt)$ is a solution to the dispersionless KdV hierarchy:
\beq\label{displesskdv}
\frac{\p v(\bt)}{\p t_k} = \frac{v(\bt)^k}{k!}  \frac{\p v(\bt)}{\p t_0}, \quad k\ge0.
\eeq
From the definitions we know that 
\beq\label{vvnbiwk}
v^{\rm NBI}(x,\br) = v^{\rm WK}\bigl(\bt^{\rm WK-NBI}(x,\br)\bigr).
\eeq
So $v^{\rm NBI}(x,\br)$ is also a solution to the dispersionless KdV hierarchy~\eqref{displesskdv}.
As we have mentioned in the Introduction, the power series $u^{\rm cBGW}(x,\bt;\e)$ defined in~\eqref{defucbgw} 
is a solution to the KdV hierarchy, so $v^{\rm cBGW}(x,\bt)$ is again a solution to~\eqref{displesskdv}.

By~\eqref{WKini23} we know that $v^{\rm WK}(X,\bigzero)=X$. It yields that  
$v^{\rm WK}(\bt)$ is the unique power series solution to the following genus~0 Euler--Lagrange equation~\cite{DW, Du96}:
\begin{equation}\label{ELeq}
	\sum_{k\geq 0}\frac{t_k}{k!} \bigl(v^{\rm WK}(\bt)\bigr)^k=v^{\rm WK}(\bt).
\end{equation}
It is also known that the genus~0 WK free energy has the expression~\cite{Du96, DZ-norm}:
\beq\label{duwk0}
\mathcal{F}^{\rm WK}_0(\bt) = \frac12\sum_{k,\ell\ge0} (t_k-\delta_{k,1}) (t_\ell-\delta_{\ell,1}) \frac{v^{\rm WK}(\bt)^{k+\ell+1}}{k! \ell! (k+\ell+1)}.
\eeq

For genus bigger than or equal to~1, we know the following structure lemma.
\begin{lem}[\cite{DVV, DZ-norm, IZ}] \label{structurewkfree}
For $g\ge1$, the genus~$g$ WK free energy has the expression:
\beq
\mathcal{F}^{\rm WK}_g(\bt) = F_g\biggl(\frac{\p v^{\rm WK}(\bt)}{\p t_0}, \dots, \frac{\p^{3g-2} v^{\rm WK}(\bt)}{\p t_0^{3g-2}}\biggr),
\eeq
where $F_g(z_1,\dots,z_{3g-2})$, $g\ge1$, are functions of $(3g-2)$ variables with 
\beq\label{F1witten}
F_1(z_1) = \frac1{24} \log z_1,
\eeq
and for $g\ge2$, $F_g(z_1,\dots,z_{3g-2})\in \mathbb{Q}[z_2,\dots,z_{3g-2}][z_1,z_1^{-1}]$. 
\end{lem}

The following lemma immediately follows from Lemma~\ref{structurewkfree} and Definition~\ref{defnbi}. 

\begin{lem}\label{structureNBIfree}
For $g\ge1$, the genus~$g$ NBI free energy has the expression:
\beq
\mathcal{F}^{\rm NBI}_g(x,\br) = F_g\biggl(\frac{\p v^{\rm NBI}(x,\br)}{\p r_0}, \dots, \frac{\p^{3g-2} v^{\rm NBI}(x,\br)}{\p r_0^{3g-2}}\biggr),
\eeq
where $F_g(z_1,\dots,z_{3g-2})$, $g\ge1$, are the same functions as those in Lemma~\ref{structurewkfree}. 
\end{lem}

In~\cite{YZQ21, YZQ23}, the following structure lemma for higher genus generalized BGW free energies is proved. 
\begin{lem}[\cite{OS, YZQ21, YZQ23}]\label{structurecBGWfree}
For $g\ge1$, the genus~$g$ generalized BGW free energy has the expression:
\beq
\mathcal{F}^{\rm cBGW}_g(x,\bt) = F_g\biggl(\frac{\p v^{\rm cBGW}(x,\bt)}{\p t_0}, \dots, \frac{\p^{3g-2} v^{\rm cBGW}(x,\bt)}{\p t_0^{3g-2}}\biggr),
\eeq
where $F_g(z_1,\dots,z_{3g-2})$, $g\ge1$, are the same functions as those in Lemma~\ref{structurewkfree}. 
\end{lem}

We are ready to prove Lemma~\ref{nbikdv}.
\begin{proof}[Proof of Lemma~\ref{nbikdv}]
Using \eqref{vvnbiwk}, \eqref{ELeq}, we find
\beq \label{vnbiini}
v^{\rm NBI}(x,r_{0},\bigzero)=\frac{\frac{x^2}{4}}{(1-r_0)^2}-\frac{x^2}{4}.
\eeq
From equations \eqref{defucbgw}, \eqref{cBGWini} we know that 
\beq \label{vcbgwini}
v^{\rm cBGW}(x,t_{0},\bigzero)=\frac{\frac{x^2}{4}}{(1-t_0)^2}.
\eeq
Lemma~\ref{nbikdv} then follows from Lemmas \ref{structureNBIfree}, \ref{structurecBGWfree}.
\end{proof}

	Before ending this section, we remark that the Virasoro constraints of $Z^{\rm NBI}(x, \br;\e)$ are explicitly given by 
	\begin{equation}\label{NBI-Virasoroeq}
		L^{\rm NBI}_{k}Z^{\rm NBI}(x,\br;\e)=0,\quad k\geq -1
	\end{equation}
	where
	\begin{align}\label{NBI-Virasoro}
			&L^{\rm NBI}_{k}=\frac{\e^2}{2}\sum_{i+j=k-1}\frac{(2i+1)!!(2j+1)!!}{2^{k+1}}\frac{\p^2}{\p r_{i}\p r_{j}}\notag
			\\
			+&\sum_{i\geq 0}\frac{(2k+2i+1)!!}{2^{k+1}(2i-1)!!}\bigl(r_i+(2i-1)!!\bigl(-\tfrac{2}{x^2}\bigr)^{i}-\delta_{i,0}\bigr)\frac{\p}{\p r_{i+k}}+\delta_{k,0}\frac{1}{16}+\delta_{k,-1}\frac{r_0^2}{2\e^2}.
	\end{align}

\section{Application}\label{S4}
	In this section, as an application of Theorem~\ref{thm-KdV-tau} we give a proof of Theorem~\ref{thm-cBGW-NBI}.
	
	According to~\cite{alexandrov2018, MMS} (cf.~\cite{YZQ21}), the generalized BGW partition function 
	$Z^{\rm cBGW}(x,\bt;\e)$ satisfies the following Virasoro constraints: 
	\begin{equation}\label{cBGW-Virasoro}
		L^{\rm cBGW}_m Z^{\rm cBGW}(x,\bt;\e)=0,\quad m\geq0,
	\end{equation}
	where
	\begin{align}
		L_{m}^{{\rm cBGW}}=\,&\sum_{i\geq 0}\frac{(2i+2m+1)!!}{2^{m+1}(2i-1)!!}(t_{i}-\delta_{i,0})\frac{\p}{\p t_{i+m}}+\frac{\e^2}{2}\sum_{i,j\ge0\atop i+j=m-1}\frac{(2i+1)!!(2j+1)!!}{2^{m+1}}\frac{\p^2}{\p t_i\p t_j}\notag\\
		&+\delta_{m,0}\Bigl(\frac{1}{16}+\frac{x^2}{8\e^2}\Bigr).
	\end{align}

	\begin{proof}[Proof of Theorem~\ref{thm-cBGW-NBI}]
		Denote
		\beq\label{ZC-def}
		\widehat Z(x,\bt;\e)=Z^{\rm NBI}\bigl(x,\bt^{\rm G}\bigl(\bt;\tfrac{x^2}{4}\bigr);\epsilon\bigr) e^{-\frac{\log2}{24}-\frac{3x^2}{8\e^2}+\frac{x^2}{4\e^2}\log(-\frac{x}{2})
			+\frac{C_1(x,\bt)}{\e^2}+\frac{C_2(x,\bt)}{\e^2}},
		\eeq
		and $\widehat{\mathcal{F}}(x,\bt;\e):=\log \widehat Z(x,\bt;\e) =: \sum_{g\ge0} \e^{2g-2} \widehat{\mathcal{F}}_{g}(x,\bt)$. Noting that 
		\beq
		M_1(x,\br;\e)=-C_1\bigl(x,\bt^{\rm G}\bigl(\br,-\tfrac{x^2}{4}\bigr)\bigr)
		,\quad
		M_2(x,\br;\e)=-C_2\bigl(x,\bt^{\rm G}\bigl(\br,-\tfrac{x^2}{4}\bigr)\bigr),
		\eeq
		we know that identity \eqref{cBGW-NBI} and \eqref{NBI-cBGW} are equivalent, 
		so it only suffices to prove \eqref{cBGW-NBI}, i.e., to prove $Z^{\rm cBGW}(x,\bt;\e)=\widehat{Z}(x,\bt;\e)$.
		
		From~\eqref{cBGWini} and Lemma~\ref{nbikdv} we know that $u^{\rm cBGW}(x,\bigzero;\e)$ and  $u^{\rm NBI}(x,\bigzero;\e)$ differ by a constant $x^2/4$.
It follows from Theorem~\ref{thm-KdV-tau} that $\widehat Z(x,\bt;\e)$
is a KdV $\tau$-function of the generalized BGW solution $u^{\rm cBGW}(x,\bt;\e)$. 
		So $\widehat{\mathcal{F}}(x,\bt;\e)-\mathcal{F}^{\rm cBGW}(x,\bt;\e)$ 
		could only be an affine linear function of~$\bt$.

By a direct calculation with the help of the second equation of \eqref{GalileanTransformationOnSolution}, we have
			\begin{align}
			L_0^{\rm cBGW}&=
			\sum_{i\ge 0}\tfrac{2i+1}{2}(t_{i}-\delta_{i,0})\frac{\p}{\p t_{i}} + \frac{1}{16} + \frac{x^2}{8\e^2} \notag\\ 
			&=\sum_{i\geq 0}
			\bigl(\tfrac{2i+1}{2}(r_i-\delta_{i,0})+\tfrac{x^2}{4}r_{i+1}\bigr)  \frac{\p}{\p r_{i}}
			+ \frac{1}{16} + \frac{x^2}{8\e^2}
			\notag\\
			&=e^{
				-\tfrac{M_1(x,\br)}{\e^2}
				-
				\tfrac{M_2(x,\br)}{\e^2}
			}
			\circ
			\bigl(\tfrac{x^2}{4} L_{-1}^{\rm NBI}+L_0^{\rm NBI}\bigr)
			\circ
			e^{
				\tfrac{M_1(x,\br)}{\e^2}
				+
				\tfrac{M_2(x,\br)}{\e^2}
			}.
			\end{align}
		Thus
		\beq
		\tfrac{x^2}{4} L_{-1}^{\rm NBI}+L_0^{\rm NBI}
		=e^{
			-\tfrac{C_1(x,\bt)}{\e^2}
			-
			\tfrac{C_2(x,\bt)}{\e^2}
		}
		\circ
		L_0^{\rm cBGW}
		\circ
		e^{
			\tfrac{C_1(x,\bt)}{\e^2}
			+
			\tfrac{C_2(x,\bt)}{\e^2}
		}.
		\eeq
		It follows that $\widehat{Z}(x,\bt;\e)$ satisfies the $m=0$ case of~\eqref{cBGW-Virasoro}, 
		so the affine linear function $\widehat{\mathcal{F}}(x,\bt;\e)-\mathcal{F}^{\rm cBGW}(x,\bt;\e)$ 
		could only be a constant that can depend on~$x,\e$.
		
	It remains to show that $\widehat{\mathcal{F}}(x,\bigzero;\e)=\mathcal{F}^{\rm cBGW}(x,\bigzero;\e)$, which is 
		equivalent to $\widehat{\mathcal{F}}_g(x,\bigzero)=\mathcal{F}^{\rm cBGW}_g(x,\bigzero)$ for all $g\ge0$.

For $g=0$, we know from \eqref{cBGW-correction}, \eqref{cBGW-structure}, \eqref{duwk0} and~\eqref{vnbiini} that 
 $\mathcal{F}^{\rm cBGW}_{0}(x,\bigzero)$ and $\widehat{\mathcal{F}}_0(x,\bigzero)$ are both equal to 
 $\frac{x^2}{4}\log(-\frac{x}{2})-\frac{3x^2}{8}$. 
 
 For genus bigger than or equal to~1, 
 Lemma~\ref{structureNBIfree}, Lemma~\ref{structurecBGWfree} and formulas \eqref{vnbiini}, \eqref{vcbgwini} we find that 
 $\mathcal{F}^{\rm cBGW}_g(x,\bigzero)$ and $\widehat{\mathcal{F}}_g(x,\bigzero)$ are equal.

	The theorem is proved.
\end{proof}

It follows from Theorem~\ref{thm-cBGW-NBI} and \eqref{Galilean-Correlator} that
for $g \ge 0$, $n\ge 1$, $k_1,\dots, k_n\ge 0$, 
	\begin{align}\label{BGW-NBI-FreeEnergy}
		\frac{\p^{n}\mathcal{F}_{g}^{\rm cBGW}(x,\bt)}{\p t_{k_1}\dots\p t_{k_n}}\bigg|_{\bt=\bigzero}
		&=\sum_{0\leq j_1\leq k_1,\dots,0\leq j_n\leq k_n}
		\frac{(\frac{x^2}{4})^{|k|-|j|}}{\prod_{i=1}^{n}(k_i-j_i)!}\frac{\p^{n}\mathcal{F}_{g}^{\rm NBI}(x,\br)}{\p r_{j_1}\cdots\p r_{j_n}}\bigg|_{\br=0}\notag\\
		&+\biggl(\frac{x^2}{4}\biggr)^{|k|+1}
		\biggl(\frac{\delta_{g,0}\delta_{n,1}}{(k_1+1)!(2k_1+1)}+\frac{\delta_{g,0}\delta_{n,2}}{k_1!k_2!(k_1+k_2+1)}\biggr),
	\end{align}
	and
		\begin{align}\label{NBI-BGW-FreeEnergy}
		&\frac{\p^{n}\mathcal{F}_{g}^{\rm NBI}(x,\br)}{\p r_{k_1}\dots\p r_{k_n}}\bigg|_{\br=\bigzero}
		=\sum_{0\leq j_1\leq k_1,\dots,0\leq j_n\leq k_n }\frac{(-\frac{x^2}{4})^{|k|-|j|}}{\prod_{i=1}^{n}(k_i-j_i)!}\frac{\p^{n}\mathcal{F}_{g}^{\rm cBGW}(x,\bt)}{\p t_{j_1}\cdots\p t_{j_n}}\bigg|_{\bt=0}\notag\\
		&+\biggl(-\frac{x^2}{4}\biggr)^{|k|+1}\biggl(\frac{\delta_{g,0}\delta_{n,1}}{(k_1+1)!}\biggl(\frac{2^{k_1+1}(k_1+1)!}{(2k_1+1)!!}-1\biggr)+\frac{\delta_{g,0}\delta_{n,2}}{k_1!k_2!(k_1+k_2+1)}\biggr).
	\end{align}

	The following corollary gives a generalization of a result of Kazarian--Norbury~\cite{KN}.
	
\begin{cor}\label{cor-cBGW-Kappa}
For $n\ge1$, $g\ge0$, and $k_1,\dots,k_n\ge0$ satisfying $|{k}|<g-1$, 
\beq\label{vanishingKN}
\int_{\overline{\mathcal{M}}_{g,n}} K_{3g-3+n-|k|}\psi^{k_1}_1\cdots\psi^{k_n}_n=0.
\eeq
For $n\geq 1$, $g\ge0$, and $k_1,\dots,k_n\geq 0$ satisfying $|k|\ge g-1$, 
\begin{align}
c_{g}(k_1,\dots,k_n)=&2^{g-1-2|k|}
			\sum_{0\leq l_1\leq k_1,\dots,0\leq l_n\leq k_n \atop g-1\le|l|\le 3g-3+n }
			\int_{\overline{\mathcal{M}}_{g,n}}K_{3g-3+n-|l|} \prod_{i=1}^{n}\frac{2^{l_i}\psi^{l_i}_{i}}{(k_i-l_i)!}\notag \\
			+&\frac{\delta_{n,1}\delta_{g,0}}{2^{k_1+1}(k_1+1)!(2k_1+1)}+\frac{\delta_{n,2}\delta_{g,0}}{2^{k_1+k_2+1}k_1!k_2!(k_1+k_2+1)}. 
			\label{cor-cBGW-Kappaeq}
\end{align}
For $n=0$ and $g\ge2$, 
\beq\label{zeropointK}
\int_{\overline{\mathcal{M}}_{g,0}} K_{3g-3} = \frac{(-1)^g B_{2g}}{2g(2g-2)}.
\eeq
\end{cor}
	
We note that formula \eqref{vanishingKN} and the $|k|=g-1$ 
case of formula \eqref{cor-cBGW-Kappaeq} 
were obtained by Kazarian and Norbury \cite{KN}, and were re-proved by Chidambaram, Garcia-Failde and Giacchetto 
in~\cite{CGFG}. 
In particular, the $|k|=g-1$ case of \eqref{cor-cBGW-Kappaeq} reads
\beq\label{BGWKappa14}
c_g(k_1,\dots,k_n) = \int_{\overline{\mathcal{M}}_{g,n}} K_{2g-2+n}\psi^{k_1}_{1}\cdots\psi^{k_n}_n,\quad g\ge 1,\, n\ge 1,\, |k|=g-1.
\eeq
We also note that formula~\eqref{zeropointK} was conjectured in~\cite{KN}, and proved in \cite{CGFG, YZQ21} 
(the proof of~\eqref{zeropointK} in~\cite{CGFG} uses topological recursion of Chekhov--Eynard--Orantin type, 
and our proof here will be similar to the one in~\cite{YZQ21} that uses the Hodge--BGW 
correspondence). 
	
	\begin{proof}[Proof of Corollary~\ref{cor-cBGW-Kappa}]
		It follows from Proposition~\ref{NBIkappaprop} and \eqref{cBGW-structure} that for $n\ge1$,
		\beq\label{cal-NBI-kappa}
		\frac{\p^n\mathcal{F}_{g}^{\rm NBI}(x,\br)}{\p r_{k_1}\dots \p r_{k_n}}\bigg|_{\br=0}
		=\delta_{|k|\le 3g-3+n}\biggl(\frac{x^2}{2}\biggr)^{|k|+1-g}
		\int_{\overline{\mathcal{M}}_{g,n}}K_{3g-3+n-|k|}
		\prod_{i=1}^{n}\psi_{i}^{k_i},
		\eeq
		\beq\label{cal-cBGW}
		\frac{\p^n\mathcal{F}_{g}^{\rm cBGW}(x,\bt)}{\p t_{j_1}\dots \p t_{j_n}}\bigg|_{\bt=0}
		=\delta_{|j|\ge g-1} x^{2|j|-2g+2}c_g(j_1,\dots,j_n).
		\eeq
		When $|k| < g-1$, substituting \eqref{cal-NBI-kappa}, \eqref{cal-cBGW} in \eqref{NBI-BGW-FreeEnergy}, we obtain~\eqref{vanishingKN}.
When $|k|\ge g-1$, substituting \eqref{cal-NBI-kappa}, \eqref{cal-cBGW} in \eqref{BGW-NBI-FreeEnergy}, we get \eqref{cor-cBGW-Kappaeq}.  
		For $n=0$, by taking ${\bf r}=\bigzero$ on both sides of~\eqref{NBI-cBGW} we obtain~\eqref{zeropointK}.
	\end{proof}		

\begin{cor}\label{cor-Kappa-cBGW}
		For $g\ge0$, $n\ge1$ and $k_1,\dots,k_n\ge0$ satisfying $g-1\le|k|\le3g-3+n$,
		\beq \label{NBI-cBGW-relations}
		\int_{\overline{\mathcal{M}}_{g,n}}K_{3g-3+n-|k|}\prod_{i=1}^n \psi_i^{k_i}
		=2^{|k|+1-g} \sum_{0\le j_1\le k_1,\cdots,0\le j_n\le k_n \atop |j|\le g-1}
		\frac{(-\tfrac14)^{|j|-|k|}}{\prod_{i=1}^{n}(k_i-j_i)!}c_g(j_1,\dots,j_n).
		\eeq
		For $g\ge0$, $n\ge1$ and $k_1,\dots,k_n\ge0$ satisfying $|k|>3g-3+n$,
		\begin{align} \label{cBGW-identities}
			&\sum_{0\le j_1\le k_1,\dots,0\le j_n\le k_n \atop |j|\le g-1}
			\frac{(-\tfrac14)^{|k|-|j|}}{\prod_{i=1}^{n}(k_i-j_i)!}c_g(j_1,\dots,j_n) \notag \\
			&=-\delta_{g,0}\delta_{n,1}
			\biggl(\frac{(-\tfrac12)^{k_1+1}}{(2k_1+1)!!}-\frac{(-\tfrac14)^{k_1+1}}{(k_1+1)!}\biggr)
			-
			\frac{\delta_{g,0}\delta_{n,2}(-\tfrac14)^{k_1+k_2+1}}{k_1!k_2!(k_1+k_2+1)}.
		\end{align}
\end{cor}
	
\begin{proof} 
We obtain \eqref{NBI-cBGW-relations} and \eqref{cBGW-identities} by substituting \eqref{cal-NBI-kappa}, \eqref{cal-cBGW} in \eqref{NBI-BGW-FreeEnergy}. 
\end{proof}

Let us now use the BGW-NBI correspondence to give a new proof of the following 
corollary.
\begin{cor}[WK-BGW correspondence \cite{YZ2023}]\label{thm3}
		We have the identity:
		\begin{equation}
			Z^{\rm cBGW}(x,\bt;\e)=Z^{\rm WK}(\bt^{\rm WK-cBGW}(x,\bt);\sqrt{-4}\e)e^{\frac{A(x,\bt)}{\e^2}},
		\end{equation}
		where
		\begin{equation}
			t_{n}^{\rm WK-cBGW}(x,\bt):=-2\sum_{k\ge0}\frac{(-1)^{n}}{k!}t_{n+k}+x\frac{(2n-1)!!}{2^n}+2\delta_{n,0}+\delta_{n,1},\quad n\ge 0,
		\end{equation}
		and
		\begin{equation}
			A(x,\bt)=\frac{1}{2}\sum_{i,j\geq 0}\frac{(t_{i}-\delta_{i,0})(t_{j}-\delta_{j,0})}{i!j!(i+j+1)}-x\sum_{i\geq 0}\frac{t_{i}-\delta_{i,0}}{i!(2i+1)}.
		\end{equation}
	\end{cor}
	
\begin{proof}
		Using the $m=0$ equation in the Virasoro constraints \eqref{cBGW-Virasoro}, we obtain 
\begin{align}\label{WK-cBGW-calc1}
			Z^{\rm cBGW}(x,\bt;\e)=\Bigl(-\frac{2}{x}\Bigr)^{\frac{1}{8}+\frac{x^2}{4\e^2}}
			Z^{\rm cBGW}(x,\bt;\e)\big|_{t_k\mapsto (-\frac{2}{x})^{2k+1}(t_k-\delta_{k,0})+\delta_{k,0},\,k\ge 0}
			.
			\end{align}
		Using Theorem~\ref{thm-cBGW-NBI} and Definition~\ref{defnbi}, we find
\begin{align}
&Z^{\rm cBGW}(x,\bt;\e)\notag\\
&=\bigl(-\tfrac{4}{x^3}\bigr)^{\frac{1}{24}}e^{\frac{A(x,\bt)}{\e^2}} 
Z^{\rm WK}(\bt;\e)\big|_{t_{n}\mapsto(-\frac{2}{x})^{2n+1}\sum_{i\ge0}\frac{1}{i!}(t_{n+i}-\delta_{n+i,0})+(2i-1)!! (-\frac{2}{x^2})^n+\delta_{n,1}}.
\label{WK-cBGW-calc4}
\end{align}
	The theorem is then proved with the help of the $k=0$ case of~\eqref{WK-Virasoro} and \eqref{WKdilaton}.
\end{proof}

From the above proof we see that, essentially speaking, 
the WK-BGW correspondence and the BGW-NBI correspondence are equivalent. 
In~\cite{YZ2023}, it was used the Hodge-BGW correspondence~\cite{YZQ21} and the Hodge-WK correspondence~\cite{YZ2023} to prove 
the WK-BGW correspondence; 
our proof here, which is based on the BGW-NBI correspondence, is a different one. 

Let $\mathbb{E}_{g,n}$ be the Hodge bundle on $\overline{\mathcal{M}}_{g,n}$, and $\Lambda(z):=\sum_{j=0}^{g}\lambda_{j}z^{j}$ the Chern polynomial of $\mathbb{E}_{g,n}$. It directly follows from Corollary~\ref{cor-cBGW-Kappa} and the Hodge-BGW correspondence~\cite{YZQ23} the 
following corollary.
	\begin{cor}
		For $g\ge0$, $n\geq 1$ and $k_1,\dots,k_n\geq 0$,
		\begin{align}
			&\sum_{0\leq l_1\leq k_1,\dots,0\leq l_n\leq k_n \atop g-1\le|l|\le 3g-3+n }
			\int_{\overline{\mathcal{M}}_{g,n}}K_{3g-3+n-|l|} \prod_{i=1}^{n}\frac{2^{l_i}\psi^{l_i}_{i}}{(k_i-l_i)!}
			+\frac{2^{k_1}\delta_{n,1}\delta_{g,0}}{(k_1+1)!(2k_1+1)}+\frac{2^{k_1+k_2}\delta_{n,2}\delta_{g,0}}{k_1!k_2!(k_1+k_2+1)}
			\notag \\
			&=\frac{(-1)^{g-1+n}2^{2|k|+g-1+n}}{(2|k|-2g+2)!\prod_{i=1}^{n}k_i!}\int_{\overline{\mathcal{M}}_{g,n+2|k|-2g+2}}\Lambda(-1)^2 \Lambda\Bigl(\frac{1}{2}\Bigr)e^{\sum_{d\ge 1}\frac{(-1)^{d-1}\kappa_{d}}{2^d d}}\prod_{i=1}^{n}\frac{1}{1+\frac{2 k_i+1}{2}\psi_i}.
		\end{align}
	\end{cor}

Similar to the BGW-NBI correspondence, 
we hope that the Galilean symmetry considered in this paper (Theorems \ref{thm-of-kdv-solution}, \ref{thm-KdV-tau})
can be useful in the study of the phase transition in two-dimensional topological gravity (cf.~\cite{Zhou}).

\section{Generalizations}\label{S5}
In this section, we generalize Theorems \ref{thm-of-kdv-solution}, \ref{thm-KdV-tau} to Frobenius manifolds.

Recall that a {\it Frobenius structure of charge~$D$} \cite{Du96} on a complex manifold~$M$ is a family of 
Frobenius algebras $(T_p M, \, e_p, \, \langle\,,\,\rangle)_{p\in M}$, 
depending holomorphically on~$p$ and satisfying the following three axioms:
\begin{itemize}
\item[{\bf A1}]  The metric $\langle\,,\,\rangle$ is flat; moreover, 
denote by~$\nabla$ the Levi--Civita connection of~$\langle\,,\,\rangle$, then it is required that 
\beq\label{flatunity511}
\nabla e = 0.
\eeq
\item[{\bf A2}] Define a $3$-tensor field~$c$ by 
$c(X,Y,Z):=\langle X\cdot Y,Z\rangle$, for $X,Y,Z$ being holomorphic vector fields on~$M$. 
Then the $4$-tensor field $\nabla c$ is required to be 
symmetric. 
\item[{\bf A3}]
There exists a holomorphic vector field~$E$ on~$M$ satisfying
\begin{align}
& \nabla\nabla E = 0, \label{linear}\\
&  [E, X\cdot Y]-[E,X]\cdot Y - X\cdot [E,Y] = X\cdot Y, \label{E2}\\
&  E \langle X,Y \rangle - \langle [E,X],Y\rangle - \langle X, [E,Y]\rangle 
= (2-D) \langle X,Y\rangle. \label{E3}
\end{align}
\end{itemize}
Here, by a Frobenius algebra $(A, \, e, \, \langle\,,\,\rangle)$ 
we mean a commutative and associative algebra~$A$ over~$\CC$
with unity~$e$ and with a symmetric, non-degenerate and invariant bilinear product $\langle\,,\,\rangle$. 
A complex manifold endowed with a Frobenius structure of charge~$D$ is  
called a {\it Frobenius manifold of charge~$D$}, with $E$ being called the {\it Euler vector field}.

Let $M$ be a Frobenius manifold. Recall that the deformed flat connection~\cite{Du96} on the Frobenius manifold~$M$, 
denoted by~$\widetilde{\nabla}(z)$, $z\in\CC$,
is a one-parameter family of flat affine connections on the tangent bundle $TM$, defined by 
\beq\label{defDubrovinconn}
\widetilde \nabla(z)_{X} Y := \nabla_X Y + z \, X \cdot Y, \quad \forall \, X,Y\in \mathcal{X}(M).
\eeq
One can extend~\cite{Du96} $\widetilde{\nabla}(z)$ to a flat connection~$\widetilde{\nabla}$ on $M\times \CC^*$ as follows:
\beq\label{defiextendedDubconn}
\widetilde \nabla_{\frac{\p}{\p z}} X := \frac{\p X}{\p z} + E \cdot X - \frac1 z \mathcal{V} X, \quad 
\widetilde \nabla_{\frac{\p}{\p z}} \frac{\p}{\p z} := 0,  \quad  \widetilde \nabla_X \frac{\p}{\p z} := 0
\eeq
for $X$ being any holomorphic vector field on $M \times \mathbb{C}^*$ with 
zero component along  $\frac{\p}{\p z}$. Here, 
\beq
\mathcal{V}:=\frac{2-d}2 {\rm id} - \nabla E.
\eeq
The connection $\widetilde{\nabla}$ is called the {\it extended deformed flat connection}.

Take  
 ${\bf v}=(v^1,\dots,v^n)$ a flat coordinate system satisfying $e=\p_\iota$. Here and below, $\p_\alpha:=\p/\p v^\alpha$. 
 We assume that 
the flat coordinates can be chosen so that the Euler vector field~$E$ has the form 
\begin{align}
&E=\sum_{\beta=1}^n \bigl(\bigl(1-\tfrac{D}2-\mu_\beta\bigr) v^\beta + r^\beta\bigr) \p_{\beta} , \label{E517} 
\end{align}
where $\mu_1,\dots,\mu_n, r^1,\dots,r^n$ are constants. 
Denote $\mu={\rm diag}(\mu_1,\dots,\mu_n)$. Note that the differential system 
for the flat section of 
$\widetilde{\nabla}$ has a Fuchsian singular point at $z=0$. 
The monodromy data at $z=0$ is given by $(\mu, R)$, where $R$ is a constant matrix (cf.~\cite{Du96, Du99, DZ-norm}).
 
 The flatness of~$\widetilde{\nabla}$ 
 ensures the existence of functions $\theta_{\alpha,m}({\bf v})$, $m\ge0$, satisfying 
\begin{align}
&
\theta_{\alpha,0}({\bf v})=v_\alpha, \qquad \alpha=1,\dots,n,\label{theta000511}\\
&
\frac{\p^2 \theta_{\gamma,m+1}({\bf v})}{\p v^\alpha\p v^\beta} = \sum_{\sigma=1}^n c^\sigma_{\alpha\beta}({\bf v}) 
\frac{\p \theta_{\gamma,m}({\bf v})}{\p v^\sigma}, \qquad \alpha,\beta=1,\dots,n, \, m\ge0,\label{theta1c}\\
&\bigl\langle \nabla \theta_\alpha({\bf v};z),  \nabla \theta_\beta({\bf v};-z) \bigr\rangle = \eta_{\alpha\beta}, \qquad \alpha,\beta=1,\dots,n,
\label{orthogtheta511}\\
&\frac{\p\theta_{\alpha,m+1}({\bf v})}{\p v^\iota} = \theta_{\alpha,m}({\bf v}), \quad \alpha=1,\dots,n, \, m\geq0, \label{thetacali512}\\
& E (\p_{\beta}(\theta_{\alpha,m}({\bf v}))) = 
(m+\mu_\alpha+\mu_\beta) \, \p_{\beta}(\theta_{\alpha,m}({\bf v})) + 
\sum_{\gamma=1}^n \sum_{k=1}^m (R_k)^\gamma_\alpha \p_{\beta} (\theta_{\gamma,m-k}({\bf v})) , \quad m\geq 0. \label{theta2c} 
\end{align}
For a given $(\mu, R)$, the functions 
$\{\theta_{\alpha,m}({\bf v})\}_{\alpha=1,\dots,n, \, m\geq0}$ 
may not be unique. A choice of $\{\theta_{\alpha,m}({\bf v})\}_{\alpha=1,\dots,n, \, m\geq0}$ satisfying \eqref{theta000511}--\eqref{theta2c} 
is called a {\it calibration}~\cite{DLYZ16, DZ-norm}. A Frobenius manifold with a fixed calibration is called {\it calibrated}. 
Below we fix a calibration. 

The {\it principal hierarchy} of a Frobenius manifold~$M$ is 
the following hierarchy of pairwise commuting evolutionary systems of PDEs of hydrodynamic-type~\cite{Du96}:
\beq\label{principalh520}
\frac{\p v^\alpha}{\p t^{\beta,m}}  = 
\sum_{\gamma=1}^n \eta^{\alpha\gamma} \p_X \biggl(\frac{\p \theta_{\beta,m+1}}{\p v^\gamma}\biggr), \quad 
\alpha,\beta=1,\dots,n, \, m\ge0.
\eeq
Since the $\p_{t^{\iota,0}}$-flow 
of the principal hierarchy~\eqref{principalh520} is just 
$\p v^\alpha/\p t^{\iota,0} = \p_X (v^\alpha)$, 
 we identify $X$ with $t^{\iota,0}$.
The flows in~\eqref{principalh520} pairwise commute~\cite{Du96},  
so we can solve them together, yielding solutions of the form
${\bf v}={\bf v}({\bf t})$, where ${\bf t}=(t^{\beta,m})_{\beta=1,\dots,n,\, m\ge0}$. 

Following Dubrovin~\cite{Du96} (cf.~\cite{DZ-norm}), define 
$\Omega_{\alpha,m_1;\beta,m_2}^{[0]}({\bf v})$, $\alpha,\beta=1,\dots,n$, $m_1,m_2\geq0$, 
for~$M$ by means of generating series as follows:
\beq\label{deftwopoint511}
\sum_{m_1,m_2\ge0} \Omega_{\alpha,m_1;\beta,m_2}^{[0]}({\bf v}) z^{m_1} w^{m_2} = 
\frac{\langle \nabla\theta_\alpha({\bf v};z), \nabla\theta_\beta({\bf v};w) \rangle-\eta_{\alpha\beta}}{z+w} , \qquad \alpha,\beta=1,\dots,n.
\eeq
We know from~\cite{Du96, DZ-norm} that 
for any solution ${\bf v}({\bf t})=(v^1({\bf t}),\dots, v^n({\bf t}))$ to the principal hierarchy, 
there exists a function $\tau_{\rm PH}({\bf t};\e)$, such that 
\beq\label{definingzerofree423}
\e^2\frac{\p^2 \log \tau_{\rm PH}({\bf t};\e)}{\p t^{\alpha,m_1} \p t^{\beta,m_2}} = \Omega_{\alpha,m_1; \beta,m_2}^{[0]}({\bf v}({\bf t})), \quad 
\alpha,\beta=1,\dots,n,\, m_1,m_2\ge0.
\eeq
We call $\tau_{\rm PH}({\bf t};\e)$ the {\it $\tau$-function of the solution~${\bf v}({\bf t})$ to the principal hierarchy~\eqref{principalh520}}. 
The logarithm $\e^2 \log \tau_{\rm PH}({\bf t};\e)=:\mathcal{F}_0({\bf t})$ is called the 
{\it genus~0 free energy of the solution~${\bf v}({\bf t})$}.

It was shown in~\cite{DZ-norm} that the principal hierarchy~\eqref{principalh520} admits the infinitesimal Galilean symmetry:
\beq\label{galileanph}
\frac{\p \tilde v^\alpha}{\p q} = \delta^\alpha_\iota + \sum_{\beta=1}^n\sum_{k\ge0} t^{\beta,k+1}\frac{\p \tilde v^\alpha}{\p t^\beta_k} , \quad \alpha=1,\dots,n,
\eeq
which can also be upgraded 
to the infinitesimal action~\cite{DZ-norm} on $\tau$-functions as follows:
\beq
\frac{\p \tilde \tau_{\rm PH}}{\p q} = \sum_{\alpha,\beta=1}^n\frac{\eta_{\alpha\beta}t^{\alpha,0}t^{\beta,0}}{2\e^2} \tilde \tau_{\rm PH} 
+ \sum_{\beta=1}^n\sum_{k\ge0} t^{\beta,k+1}\frac{\p \tilde \tau_{\rm PH}}{\p t^\beta_k} .
\eeq

\begin{thm}\label{thmphg}
Let $M$ be a calibrated Frobenius manifold, and 
$v(\bt)$ an arbitrary solution to the principal hierarchy of~$M$. Define $\tilde {\bf v}(\bt;q)=(\tilde v^1(\bt;q),\dots,\tilde v^n(\bt;q))$ by
\beq\label{tildevggp}
\tilde v^\alpha(\bt;q) = v^\alpha\bigl(\bt^{\rm G}(\bt;q)\bigr) + q\delta^{\alpha}_{\iota},
\eeq
where 
\begin{equation}\label{tgfm}
		(t^{\rm G})^{\alpha,n}(\bt;q)=\sum_{k\geq 0}\frac{q^k}{k!}t^{\alpha,n+k}.
\end{equation}
Then $\tilde {\bf v}(\bt;q)$ is a solution to the principal hierarchy~\eqref{principalh520} for any~$q$. Moreover, 
let $\tau_{\rm PH}({\bf t};\e)$ be a $\tau$-function of the solution $v(\bt)$. Then 
$\tilde \tau_{\rm PH}(\bt;\e;q)$ defined by 
\beq\label{tildetauphgp}
\tilde \tau_{\rm PH}(\bt;\e;q) = \tau_{\rm PH}\bigl(\bt^{\rm G}(\bt;q);\e)\bigr) e^{\frac{g(\bt;q)}{\e^2}}	
\eeq
with 
\beq\label{defgfm}
g(\bt;q):=\frac{1}{2}\sum_{\alpha,\beta=0}^n \sum_{i,j\geq 0}\frac{q^{i+j+1}}{i+j+1}\frac{t^{\alpha,i}}{i!}\frac{t^{\beta,j}}{j!}\eta_{\alpha\beta}
\eeq
is a $\tau$-function of the solution~$\tilde v(\bt;q)$.
\end{thm}
\begin{proof}
Similar to the proof of Proposition~\ref{prop-galilean-solution} one can 
prove by a straightforward verification that $\tilde{\bf v}(\bt;q)$ defined by~\eqref{tildevggp} 
satisfies~\eqref{galileanph}.  
Since equation~\eqref{galileanph} gives an infinitesimal symmetry of the principal hierarchy~\eqref{principalh520}, 
the statement that $\tilde {\bf v}(\bt;q)$ is a solution to the principal hierarchy follows.
Using \eqref{tildetauphgp}, \eqref{definingzerofree423} and \eqref{tgfm}, we have
		\begin{align}
			\e^2\frac{\p^2 \log \tilde{\tau}_{\rm PH}(\bt;\e;q)}{\p t^{\alpha,k_1} \p t^{\beta,k_2}}
			=& \sum_{0\le m_1 \le k_1, 0\le m_2\le k_2}\frac{q^{k_1+k_2-m_1-m_2}}{(k_1-m_1)!(k_2-m_2)!}
			\Omega_{\alpha,m_1;\beta,m_2}^{[0]}\bigl({\bf v}\bigl(\bt^{\rm G}(\bt;q)\bigr)\bigr)\notag \\ &+\frac{q^{k_1+k_2+1}}{k_1!k_2!(k_1+k_2+1)}\eta_{\alpha \beta}.
		\end{align}
Then, using \eqref{deftwopoint511}, we obtain that 
\begin{align}
&\e^2 \sum_{k_1,k_2\ge 0} \frac{\p^2 \log \tilde{\tau}_{\rm PH}(\bt;\e;q)}{\p t^{\alpha,k_1} \p t^{\beta,k_2}} z^{k_1} w^{k_2} \notag \\
&=e^{qz+qw}\frac{\langle \nabla \theta_\alpha({\bf v}(\bt^{\rm G}(\bt;q));z), \nabla \theta_\beta({\bf v}(\bt^{\rm G}(\bt;q));w)\rangle-\eta_{\alpha \beta}}{z+w}+\frac{(e^{qz+qw}-1)\eta_{\alpha \beta}}{z+w}
			\notag \\
			&=\frac{\langle \nabla \theta_\alpha(\tilde{\bf v}(\bt;q);z), \nabla \theta_\beta(\tilde{\bf v}(\bt;q);w)\rangle-\eta_{\alpha \beta}}{z+w}.
\end{align}
The theorem is proved.	
\end{proof}

Actually, the principal hierarchy admits an infinite family of Virasoro symmetries with the above Galilean symmetry being one of 
them. However, as it was observed by Dubrovin and Zhang~\cite{DZ-norm}, differently from the Galilean symmetry, 
the actions of most of other Virasoro symmetries on 
tau-functions for the principal hierarchy are not linear. Nevertheless, 
for the case when the Frobenius manifold is {\it semisimple}, 
Dubrovin and Zhang~\cite{DZ-norm} introduced the notion of linearization of Virasoro symmetries and 
constructed higher genus free energies $F_g({\bf v}, {\bf v}_1,\dots, {\bf v}_{3g-2})$, $g\ge1$, as solutions to 
the {\it loop equation}~\cite{DZ-norm}. 
A particular perturbation of the principal hierarchy, called the {\it integrable hierarchy of topological type}, 
is then constructed as the substitution of the quasi-trivial transformation 
\beq\label{quasitmap}
v^\alpha \mapsto u^\alpha= v^\alpha+ \sum_{\beta=1}^n \eta^{\alpha\beta} \p^2( F_g({\bf v}, {\bf v}_1,\dots, {\bf v}_{3g-2})), \quad \alpha=1,\dots,n,
\eeq
into the principal hierarchy. Here, $\p=\p_X=\sum_{\gamma=1}^n\sum_{k\ge0} v^\gamma_{k+1} \p/\p v^\gamma_k$. 
This particular perturbation, now often called the {\it Dubrovin--Zhang hierarchy of~$M$}, has the form \cite{DZ-norm}
\beq\label{DZhierarchy}
\frac{\p u^\alpha}{\p t^{\beta,m}}  = 
\sum_{\gamma=1}^n P^{\alpha\gamma} \biggl(\frac{\delta H_{\beta,m}}{\delta u^\gamma(X)}\biggr), \quad 
\alpha,\beta=1,\dots,n, \, m\ge0,
\eeq
where $P$ is a Hamiltonian operator (which allows a power-series dependence in~$\e$), $\delta/\delta u^\gamma(X)$ 
denote the variational derivatives~\cite{DZ-norm}, $H_{\beta,m}$ are Hamiltonians, and
$X=t^{\iota,0}$. 

Denote ${\bf u}=(u^1,\dots,u^n)$ and 
define  $\Omega_{\alpha,m_1;\beta,m_2}({\bf u}, {\bf u}_1,\dots;\e)$
as the substitutions of the inverse of the 
quasi-trivial map~\eqref{quasitmap} in 
$$\Omega_{\alpha,m_1;\beta,m_2}^{[0]}({\bf v})+\sum_{g\ge1} \e^{2g} \p^2(F_g({\bf v}, {\bf v}_1,\dots, {\bf v}_{3g-2})). $$
Here, $\alpha,\beta=1,\dots,n$, $m_1,m_2\ge0$.
Polynomiality of coefficients in $P^{\alpha\gamma}$, in the densities of $H_{\beta,m}$, and in $\Omega_{\alpha,m_1;\beta,m_2}({\bf u}, {\bf u}_1,\dots;\e)$
was proved in~\cite{BPS}.
For any arbitrary solution ${\bf u}(\bt;\e)$ to the Dubrovin--Zhang hierarchy~\eqref{DZhierarchy}, 
there exists~\cite{BPS, DZ-norm} a function $\tau_{\rm DZ}(\bt;\e)$ such that 
\beq\label{definingtauDZ}
\e^2\frac{\p^2 \log \tau_{\rm DZ}({\bf t};\e)}{\p t^{\alpha,m_1} \p t^{\beta,m_2}} 
= \Omega_{\alpha,m_1; \beta,m_2}\biggl({\bf u}({\bf t};\e), \frac{\p {\bf u}(\bt;\e)}{\p X}, \dots;\e\biggr).
\eeq
The function $\tau_{\rm DZ}(\bt;\e)$ is called the {\it $\tau$-function of the solution ${\bf u}(\bt;\e)$}.

By definition the Dubrovin--Zhang hierarchy~\eqref{DZhierarchy} admits the following infinitesimal Galilean symmetry~\cite{DZ-norm}:
\beq\label{infGu}
\frac{\p \tilde u^\alpha}{\p q} = \delta^\alpha_\iota + \sum_{\beta=1}^n\sum_{k\ge0} t^{\beta,k+1}\frac{\p \tilde u^\alpha}{\p t^\beta_k} , \quad \alpha=1,\dots,n.
\eeq	
Let $\tau_{\rm DZ}(\bt;\e)$ be any $\tau$-function for the Dubrovin--Zhang hierarchy of~$M$.  Then, according to the linearization 
of Virasoro symmetries,  the infinitesimal Galilean symmetry~\eqref{infGu} is upgraded as follows:
\beq
\frac{\p \tilde \tau_{\rm DZ}}{\p q} = \sum_{\alpha,\beta=1}^n\frac{\eta_{\alpha\beta}t^{\alpha,0}t^{\beta,0}}{2\e^2} \tilde \tau_{\rm DZ} 
+ \sum_{\beta=1}^n\sum_{k\ge0} t^{\beta,k+1}\frac{\p \tilde \tau_{\rm DZ}}{\p t^{\beta,k}} .
\eeq

We have the following theorem generalizing Theorems \ref{thm-of-kdv-solution}, \ref{thm-KdV-tau}. 
\begin{thm}\label{DZGthm}
	Let $M$ be a semisimple calibrated Frobenius manifold, 
and let ${\bf u}(\bt;\e)$ be an arbitrary solution to the Dubrovin--Zhang hierarchy of~$M$. 
Define $\tilde {\bf u}(\bt;\e;q)=(\tilde u^1(\bt;\e;q),\dots,\tilde u^n(\bt;\e;q))$ by
\beq\label{tildeuggp}
\tilde u^\alpha(\bt;\e;q) = u^\alpha\bigl(\bt^{\rm G}(\bt;q);\e\bigr) + q\delta^{\alpha}_{\iota}, \quad \alpha=1,\dots,n,
\eeq
where $\bt^{\rm G}(\bt;q)$ is defined in~\eqref{tgfm}.
Then $\tilde {\bf u}(\bt;\e;q)$ is a solution to the Dubrovin--Zhang hierarchy of~$M$ for any~$q$. Moreover, 
let $\tau_{\rm DZ}({\bf t};\e)$ be a $\tau$-function of the solution ${\bf u}(\bt;\e)$. Then 
$\tilde \tau_{\rm DZ}(\bt;\e;q)$ defined by 
\beq\label{deftildetaudz}
\tilde \tau_{\rm DZ}(\bt;\e;q) = \tau_{\rm DZ}\bigl(\bt^{\rm G}(\bt;q);\e\bigr) e^{\frac{g(\bt;q)}{\e^2}}, \eeq
is a $\tau$-function of the solution~$\tilde {\bf u}(\bt;\e;q)$. Here $g(\bt;q)$ is the quadratic function defined in~\eqref{defgfm}.
\end{thm}
\begin{proof}
Again, from a straightforward verification we know that $\tilde{\bf u}(\bt;q)$ defined by~\eqref{tildeuggp} 
satisfies~\eqref{infGu}.  Since 
 the Dubrovin--Zhang hierarchy~\eqref{DZhierarchy} 
admits the infinitesimal Galilean symmetry~\eqref{infGu}, we know that  
$\tilde {\bf u}(\bt;\e;q)$, for any~$q$, is a solution to the Dubrovin--Zhang hierarchy~\eqref{DZhierarchy}. 
Using \eqref{deftildetaudz}, \eqref{tildevggp}, Theorem~\ref{thmphg}, the definition of 
$\Omega_{\alpha,m_1;\beta,m_2}({\bf u}, {\bf u}_1,\dots;\e)$ and 
the well-known fact that $F_g$, $g\ge1$,
do not explicitly contain~$v^\iota$ we find the validity of the second statement.
\end{proof}

When $M$ is the one-dimensional Frobenius manifold with the potential $F=v^3/6$, according to the 
Witten--Kontsevich theorem the Dubrovin--Zhang hierarchy of~$M$ coincides with the KdV hierarchy. 
Noticing that the definition of $\tau$-functions for the KdV hierarchy using the bilinear equations 
and the one using~\eqref{definingtauDZ} are equivalent~\cite{BDY}, 
we know that the above Theorem~\ref{DZGthm} becomes Theorems \ref{thm-of-kdv-solution}, \ref{thm-KdV-tau}.
Applications of Theorem~\ref{DZGthm} will be considered elsewhere. 

In a similar way, we have the following theorem.
\begin{thm}\label{DZHodgeGthm}
	Let $M$ be a semisimple calibrated Frobenius manifold, 
and let ${\bf w}(\bt;\boldsymbol{\sigma};\e)$ be an arbitrary solution to the Hodge hierarchy~\cite{DLYZ16} associated with~$M$, 
where $\boldsymbol{\sigma}=(\sigma_1,\sigma_3,\sigma_5,\dots)$ are Hodge parameters~\cite{DLYZ16}. 
Define 
\beq\label{tildeuggphodge}
\tilde w^\alpha(\bt;\boldsymbol{\sigma};\e;q) = w^\alpha\bigl(\bt^{\rm G}(\bt;q);\boldsymbol{\sigma};\e\bigr) + q\delta^{\alpha}_{\iota}, \quad \alpha=1,\dots,n,
\eeq
where $\bt^{\rm G}(\bt;q)$ is defined in~\eqref{tgfm}.
Then $\tilde {\bf w}(\bt;\boldsymbol{\sigma};\e;q)$ is a solution to the Hodge hierarchy of~$M$ for any~$q$. Moreover, 
let $\tau_{\rm H}({\bf t};\boldsymbol{\sigma};\e)$ be a $\tau$-function~\cite{DLYZ16} of the solution ${\bf w}(\bt;\boldsymbol{\sigma};\e)$. Then 
$\tilde \tau_{\rm H}(\bt;\boldsymbol{\sigma};\e;q)$ defined by 
\beq\label{deftildetaudzhodge}
\tilde \tau_{\rm H}(\bt;\boldsymbol{\sigma};\e;q) = \tau_{\rm H}\bigl(\bt^{\rm G}(\bt;q);\boldsymbol{\sigma};\e\bigr) e^{q\frac{\sigma_1}{24}+\frac{g(\bt;q)}{\e^2}}, \eeq
is a $\tau$-function of the solution~$\tilde {\bf w}(\bt;\boldsymbol{\sigma};\e;q)$. Here $g(\bt;q)$ is the quadratic function defined in~\eqref{defgfm}.
\end{thm}

\smallskip

\noindent {\bf Acknowledgements.} 
The work is partially supported by NSFC No. 12371254 and CAS No. YSBR-032.

\medskip
\medskip
\medskip

\noindent School of Mathematical Sciences, University of Science and Technology of China,

\noindent Hefei 230026, P.R. China 

\noindent xjh\_020403@mail.ustc.edu.cn,   diyang@ustc.edu.cn

\end{document}